\documentclass[numbook,envcountsame,smallextended,onecolumn]{svjour3mod}

\usepackage{ifthen}

\makeatletter
\newif\if@book
\@bookfalse   
\relax
\let\ifbook\if@book

\makeatletter
\newif\if@arxiv
\@arxivtrue   
\relax
\let\ifarxiv\if@arxiv

\ifbook\arxivfalse\else\fi

\smartqed  

\usepackage{pdflscape}

\usepackage{mathptmx}       
\usepackage{helvet}         
\usepackage{courier}        
\usepackage{type1cm}        
\usepackage{makeidx}         
\usepackage{graphicx}        
\usepackage{multicol}        
\usepackage[bottom]{footmisc}

\usepackage[
    backend=biber,
    bibstyle=alphabetic,
    citestyle=alphabetic,
    sorting=nty,
    natbib=false,
    isbn=false,
    uniquename=false,
    firstinits=true,
    backrefstyle=two,
    abbreviate=true,
    bibwarn=false,
    url=false, 
    doi=false,
    safeinputenc,
    bibencoding=auto,
    maxnames=5,
    mincrossrefs=99,
    maxalphanames=5
]{biblatex}
\AtEveryBibitem{%
  \clearname{editor}%
  \clearlist{language}%
  \clearfield{month}%
  \clearlist{location}%
  \clearfield{pagetotal}%
}

\renewbibmacro{in:}{%
  \ifentrytype{article}{}{\printtext{\bibstring{in}\intitlepunct}}}

\DeclareFieldFormat{titlecase}{\MakeTitleCase{#1}}

\newrobustcmd{\MakeTitleCase}[1]{%
  \ifthenelse{\ifcurrentfield{booktitle}\OR\ifcurrentfield{booksubtitle}%
    \OR\ifcurrentfield{maintitle}\OR\ifcurrentfield{mainsubtitle}%
    \OR\ifcurrentfield{journaltitle}\OR\ifcurrentfield{journalsubtitle}%
    \OR\ifcurrentfield{issuetitle}\OR\ifcurrentfield{issuesubtitle}%
    \OR\ifentrytype{book}\OR\ifentrytype{mvbook}\OR\ifentrytype{bookinbook}%
    \OR\ifentrytype{booklet}\OR\ifentrytype{suppbook}%
    \OR\ifentrytype{collection}\OR\ifentrytype{mvcollection}%
    \OR\ifentrytype{suppcollection}\OR\ifentrytype{manual}%
    \OR\ifentrytype{periodical}\OR\ifentrytype{suppperiodical}%
    \OR\ifentrytype{proceedings}\OR\ifentrytype{mvproceedings}%
    \OR\ifentrytype{reference}\OR\ifentrytype{mvreference}%
    \OR\ifentrytype{report}\OR\ifentrytype{thesis}\OR\ifentrytype{phdthesis}}
    {#1}
    {\MakeSentenceCase{#1}}}

\renewcommand{\bibfont}{\small}
\newcommand{\takeout}[1]{}

\ifbook\else
\usepackage{a4wide}
\fi

\usepackage{amsfonts}
\usepackage{amsmath}
\usepackage{amssymb}
\usepackage{stmaryrd}
\usepackage{hyperref}
\usepackage{array}
\usepackage{multirow}
\usepackage{color}
\usepackage{mathpartir}
\usepackage{mathrsfs}
\usepackage{lscape}
\usepackage{
pgfrcs}
\usepackage{
pgf}
\usepackage{
tikz}
\usetikzlibrary{%
  shapes.misc,
  shapes.geometric,
  positioning,
  shadows%
}
\usepackage{turnstile}
\usepackage[all]{xy}
\xyoption{2cell}
\xyoption{curve}
\UseTwocells
\SelectTips{cm}{}
\usepackage{stmaryrd}

\spnewtheorem{thm}[theorem]{Theorem}{\bfseries}{\itshape}
\spnewtheorem{cor}[theorem]{Corollary}{\bfseries}{\itshape}
\spnewtheorem{lem}[theorem]{Lemma}{\bfseries}{\itshape}
\spnewtheorem{lemdefn}[theorem]{Lemma and Definition}{\bfseries}{\itshape}
\spnewtheorem{prop}[theorem]{Proposition}{\bfseries}{\itshape}
\spnewtheorem{fact}[theorem]{Fact}{\bfseries}{\itshape}
\spnewtheorem{defn}[theorem]{Definition}{\bfseries}{\upshape}
\spnewtheorem{rem}[theorem]{Remark}{\bfseries}{\upshape}
\spnewtheorem{expl}[theorem]{Example}{\bfseries}{\upshape}
\spnewtheorem{thmdefn}[theorem]{Theorem and Definition}{\bfseries}{\itshape}
\spnewtheorem{propdefn}[theorem]{Proposition and Definition}{\bfseries}{\itshape}
\spnewtheorem{assumption}[theorem]{Assumption}{\bfseries}{\upshape}
\spnewtheorem{algorithm}[theorem]{Algorithm}{\bfseries}{\upshape}
\spnewtheorem{conv}[theorem]{Convention}{\bfseries}{\upshape}
\spnewtheorem{notn}[theorem]{Notation}{\bfseries}{\upshape}
\spnewtheorem{probl}[theorem]{Problem}{\bfseries}{\upshape}

\spnewtheorem{defprop}[theorem]{Definition and Proposition}{\bfseries}{\upshape}

\spnewtheorem{defremark}[theorem]{Definition and Remark}{\bfseries}{\upshape}

\usepackage{wasysym} 
\usepackage{pigpen}  

\DeclareMathAlphabet{\mathsfit}{\encodingdefault}{\sfdefault}{m}{sl}
\SetMathAlphabet{\mathsfit}{bold}{\encodingdefault}{\sfdefault}{bx}{sl}

\newcommand{\algcl}[1]{\ensuremath{\mathsf{#1}}}
\newcommand{\algvr}[1]{\ensuremath{\mathsfit{#1}}}

\newcommand{\HA}{\algcl{HA}}

\newcommand{\bi}{\mathbin{{-}\!*}}

\newcommand{\alo}[1]{\ensuremath{\mathbb{#1}}}

\newcommand{\Ho}{\alo{H}}
\newcommand{\Po}{\alo{P}}
\newcommand{\So}{\alo{S}}
\newcommand{\Vo}{\alo{V}}

\newcommand{\lL}{\mathcal{L}}

\newcommand{\nat}{\mathbb{N}}
\newcommand{\inte}{\mathbb{Z}}

\renewcommand{\qed}{\hfill\ensuremath{\dashv}}

\def\refeq#1{(\ref{#1})}  

\newcommand{\Ppm}[1]{\m{#1}}
\newcommand{\ppm}{\cdot} 
\newcommand{\ppmu}{\mathsf{e}} 
\newcommand{\Cmx}[1]{{#1}^+}

\newcommand{\ppmsb}{\sqsubseteq}
\newcommand{\ppmsp}{\sqsupseteq}

\newcommand{\BI}{\algcl{BI}}
\newcommand{\GBI}{\algcl{GBI}}
\newcommand{\BBI}{\algcl{BBI}}

\newcommand{\BA}{\algcl{BA}}
\newcommand{\GA}{\algcl{GA}}

\newcommand{\ldd}{\backslash\!\!\backslash}
\newcommand{\rdd}{/\!\!/}
\newcommand{\m}{\mathfrak}
\newcommand{\Lu}{\m L\!\!\!\scriptsize\raise1.25pt\hbox{-}\,}
\newcommand{\Lv}{\algcl{L}\!\!\!\scriptsize\raise1.25pt\hbox{-}\,}

\newcommand{\omt}{\curlywedge}

\newcommand{\bro}{\textup{(}}
\newcommand{\brc}{\textup{)\,}}

\newcommand{\cf}[1]{\text{#1}}

\newcommand{\intu}[1]{\,{#1}_\mathsf{s}\,}
\newcommand{\intur}[1]{\,{#1}_\mathsf{s}\,}
\newcommand{\ipmsb}{\intur{\preceq}}
\newcommand{\ipco}{\intur{\succeq}}

 \newcommand{\CM}{C_\Ppm{M}}

\usepackage{tikz}
\usepackage[final,inline]{fixme} 
\FXRegisterAuthor{tl}{atl}{Tadeusz} 
\FXRegisterAuthor{pj}{apj}{Peter} 

\newcommand{\tlnt}[1]{\tlnote[inline,marginclue]{#1}}
\newcommand{\pjnt}[1]{\pjnote[inline,marginclue]{#1}}

\ifarxiv\else
\usepackage{hyperref}
\fi

\usepackage{enumitem}
\newlist{upbroman}{enumerate}{1}
\setlist[upbroman]{label=\bro\roman*\brc}

\newlist{esakia}{enumerate}{10}
\setlist[esakia]{label=\arabic*.,resume}

\newcommand{\upstx}{\mathit{Up}}
\newcommand{\ups}[1]{\upstx(#1)}

\newcommand{\fmal}{\mathfrak{Fm}}

\begin{document}

\title{An Algebraic Glimpse at Bunched Implications and Separation Logic}
\author{Peter Jipsen and Tadeusz Litak} 
\institute{ Peter Jipsen \at
School of Computational Sciences, Chapman University \newline
\email{jipsen@chapman.edu}
\and
Tadeusz Litak \at Informatik 8, FAU Erlangen-N\"{u}rnberg \newline 
\email{tadeusz.litak@fau.de}}

\date{Submitted September 2017. Accepted version prepared for publication as of \today}

\maketitle

\begin{abstract}{}
We overview the logic of Bunched Implications (BI) and Separation Logic (SL) from a perspective inspired by Hiroakira Ono's algebraic approach to substructural logics. We propose  \emph{generalized BI algebras} (\emph{GBI-algebras}) as a common framework for algebras arising via ``declarative resource reading'', 
 intuitionistic generalizations of relation algebras and arrow logics and the distributive Lambek calculus with intuitionistic implication. Apart from existing models of BI (in particular, heap models and effect algebras), we also cover models arising from weakening relations,  formal languages  or more fine-grained treatment of labelled trees and semistructured data. After briefly discussing the lattice of subvarieties of \algcl{GBI}, we present a suitable duality for \algcl{GBI} along the lines of Esakia and Priestley and an algebraic proof of cut elimination in the setting of residuated frames of Galatos and Jipsen. We also show how the algebraic approach allows generic results on decidability, both positive and negative ones. In the final part of the paper, we gently introduce the substructural audience to some theory behind state-of-art tools, culminating with an algebraic and proof-theoretic presentation of \emph{\bro bi-\brc abduction}.  
\end{abstract}

\ifbook\else
\tableofcontents
\fi



\bigskip

{
\footnotesize 
\noindent
Given the size of this paper and possibly divergent interests of its readers, we tried to ensure a 
 degree of independence between sections, even at the expense of certain redundancy.}

\section{Introduction} \label{sec:intro}

One of the major themes  in the distinguished career of Hiroakira Ono, especially since the mid-1990s, has been the use of algebraic methods for substructural purposes. The monograph \cite{GalatosJKO07} which inspired the title of this overview is a monument to this approach. While for ``standard'' substructural logics we already have an impressive body of work, we believe the time has come to promote this approach  for  more applied systems, such as generalizations and extensions of the \emph{logic of Bunched Implications} (BI).

The shortest description of BI  is: take the (commutative) propositional substructural 
  signature $\wedge$, $\vee$, $\top$, $\bot$, $*$, $\bi$, $1$ and add  the Heyting implication $\to$ adjoint to the additive conjunction $\wedge$, immediately forcing the lattice structure to be distributive. The system was explicitly introduced by O'Hearn and Pym \cite{OHearnP99:jsl,OHearn99:tlca,Pym99:lics,Pym02:book}. 
 In addition, Pym et al. \cite{PymOHY04:tcs}, Brotherston et al. \cite{BrotherstonC10:lmcs,BrotherstonK14:jacm}, O'Hearn \cite{OHearn12:nato}, Demri and Deters \cite{DemriD15:jancl} and this paper (\S\S~\ref{ModelsBI}--\ref{ModelsGBI} and \S\S~\ref{SeparationLogic}--\ref{sec:algsep}) present CS 
 applications of BI and related formalisms such as \emph{separation logic}.
  Perhaps the most important motivation 
   can be summarized in one phrase: modular reasoning about the use of \emph{shared mutable data structures} \cite{Reynolds00:intuitionistic,Reynolds02:lics}, i.e.,
\begin{quote}
structures where an updatable field can be referenced from more
than one point. \cite{Reynolds02:lics}
\end{quote}
In  a narrow sense, this applies to heap mutation, pointer aliasing and (de)allo\-cation: in short, dynamic memory management. 
  Nevertheless, BI was originally proposed   \emph{in the context of a broader investigation of ``resource modelling''} \cite{PymOHY04:tcs} (cf. \cite{OHearnPVH15}), with its ``declarative'' approach to resource contrasting with that of ``proofs-as-actions'' in linear logic. 
 Unfortunately, we can only briefly discuss the most challenging appli\-cation---i.e., shared-memory concurrency---in the concluding \S~\ref{sec:concur}.

 Despite the rather recent history of BI, several of the key ideas have been around for a long time and some of them can be immediately recognized by substructural logicians. Moreover, the authors of pioneering BI papers made no secret of these substructural origins. 
 Indeed, according to O'Hearn and Pym \cite[\S~9]{Pym99:lics}, \cite{OHearnP99:jsl}, \cite[Ch. 1]{Pym02:book} not only should BI be seen as a relevance logic, in fact up to minor syntactic details, an instance of Belnap's \emph{scheme of display logic}  \cite{Belnap82:jpl}, but also the terminology \emph{bunches} (for a structured term of formulas rather than a sequence of formulas) from which the very name of BI 
 derives, comes from Dunn's work on a sequent calculus for the relevance logic R \cite{Dunn75:chapter,Belnap82:jpl,Read88:book}; cf. also the work of Mints \cite{Mints1976}.
  
  On the mathematical side of things, another important early reference \cite{Day70:lnm} shows that a monoidal (not necessarily closed) structure on a category induces a monoidal closed structure on the corresponding category of Set-presheaves. Together with Urquhart's work on the semantics of multiplicative intuitionistic linear logic \cite{Urquhart72:jsl}, this  motivated \emph{total monoid semantics} of O'Hearn and Pym \cite{OHearnP99:jsl,Pym99:lics,Pym02:book}, known to be incomplete in the presence of $\bot$ \cite[Prop. 4.8]{Pym02:book}, \cite[Prop. 6]{PymOHY04:tcs}.

On the CS side of things, Reynolds \cite{Reynolds00:intuitionistic,Reynolds02:lics} claims that the earliest motivation for \emph{separating conjunction} is implicit in Burstall's \cite{Burstall72:mi} early idea of a \emph{distinct nonrepeating tree system}. To be more precise, rather than inspiring directly the earliest BI references, Burstall's work  inspired Reynolds' development of a  language of \emph{Hoare triples} (\S~\ref{SeparationLogic}) for programs involving shared mutable data structures.   Ishtiaq and O'Hearn \cite{IshtiaqOH01:popl} clarify the connection between that language, which came to be known as \emph{separation logic} (\S\S~\ref{SeparationLogic}--\ref{sec:algsep}), and BI as the core of its assertional part; let us note here that this paper is very explicit about the substructural character of BI. 

We hope  to dispel here whatever remains of a prejudice which seemed to linger in the early days of BI, best illustrated by  the following claim:

\begin{quote}
 We are not looking
for an algebraic semantics here, where one takes (say) a Heyting
algebra with
enough structure to model the multiplicatives; this would just be a collapsed version of the DCC semantics, and would not be very informative. \cite[\S~5]{OHearnP99:jsl}
\end{quote}
It is worth noting that even a later influential paper by those authors \cite{PymOHY04:tcs} does find the need for an algebraic  unifying framework 
 when the focus is on 
  \emph{theoremhood} rather than, say, the structure of proofs.

 More importantly, it is all too rarely mentioned that the core idea of BI---that of  allowing not only multiplicative \emph{and} additive conjunctions, but also their corresponding adjoint implications
 ---\emph{has} been seriously studied elsewhere, especially in the boolean setting. 
  This was happening mostly in the 1990's for two convergent reasons. On the one hand, such research was motivated by generalizations of 
  relation algebras, like in the PhD dissertation of the first author \cite{Jipsen92}. On the other hand, it was inspired by a \emph{dynamic trend} in formal semantics of natural language and in information processing. Volumes of collected papers from the period 
   \cite{marx:arro96} illustrate how fruitful this convergence was. The separation logic and BI  communities seem largely unaware of this body of work and, at least in some cases, have proved overlapping results. 
   Let us point out here just two examples, both of which can be traced back to the above-mentioned 1996 collection  \cite{marx:arro96} and which will be presented in more detail in  \S~\ref{UndecEq}: 
\begin{itemize}
\item 
Algebraic results regarding 
 (un)decidability for arrow logics \cite{AndrekaKNSS96:csli,KuruczNSS95:jolli} overlap with results for boolean BI  announced almost two decades later. 
\item Another example is provided by \emph{logic for layered graphs} \cite{CollinsonMDP14:jlc}, which essentially reinvents  \emph{conjugated arrow logic} \cite{Mikulas96:csli} (see Remark \ref{rem:nonassoc}).
\end{itemize}
To unify these convergent lines of research, in \S~\ref{LogicAlgebra} 
 we take as our base variety the class  $\algcl{GBI}$ of \emph{generalized} BI algebras, dropping the assumption of commutativity. In other words, we replace $*$ and $\bi$ in the substructural reduct by $\cdot, \backslash$ and $/$ (notation from the monograph \cite{GalatosJKO07}): clearly, one needs two residuals instead of one.
  This is not an uncommon step in the substructural setting; it is enough to recall how BL algebras were generalized to GBL algebras \cite{JipsenTsinakis02}. 
    Here, 
     we cover not only previously known commutative models of BI (\S~\ref{ModelsBI}), but also 
     weakening relations (\S~\ref{sec:weakening}) or  formal languages (\S~\ref{sec:language}); moreover, we improve the treatment of labelled trees and semistructured data (\S~\ref{sec:semi}).\footnote{Let us note that dropping \emph{associativity} of multiplicative conjunction has also been considered from all three perspectives, i.e., that of arrow logic, that of substructural logic, and most recently that of BI and resource reasoning. See Remark \ref{rem:nonassoc} 
for more information.}




 \S~\ref{Subvarieties} discusses systematically subvarieties of $\algcl{GBI}$, both those arising as generalizations of relation algebras and those  obtained by adopting subvarieties of residuated lattices/FL-algebras 
  \cite[\S~3.5]{GalatosJKO07} to our signature. In \S~\ref{RelSem}, we present 
   a systematic approach to semantics via 
   Esakia- and Priestley-style duality. 
  \S~\ref{Decidability} provides an overview of decidability and undecidability results for (quasi)equational theories. 
    \S~\ref{ProofTheory} discusses an algebraic take on proof theory of (G)BI,  in particular cut elimination; 
     given the contributions of Hiroakira Ono to the literature on algebraic cut elimination \cite{BelardinelliJO04,GalatosJKO07,GalatosO10:apal}, this  seems particularly natural material to present here.
 \S~\ref{SeparationLogic} is a gentle introduction to separation logic (SL) for an algebraically oriented audience. \S~\ref{sec:sltheory} reexamines the themes of  \S~\ref{Decidability} and \S~\ref{ProofTheory} from the perspective of SL, sketching a suitable substructural proof system. In \S~\ref{sec:biabduction}, we use the framework developed in earlier sections to suggest an algebraic and proof-theoretic approach to \emph{\bro bi-\brc abduction}. 
  Finally, \S~\ref{sec:algsep} provides a short glimpse at applications and developments we could not cover in detail in this overview, especially state-of-the art tools or the treatment of concurrency. 
 

\section{Logic and Algebra} \label{LogicAlgebra}

In this section, we discuss the basic algebraic setup 
   and connection to Hilbert-style calculus. The main goal here is to present the class $\algcl{GBI}$ of \emph{generalized} (non-commuta\-tive) BI algebras as a general framework for our paper. However, even when it comes to standard, commutative BI-algebras, we believe the Hilbert-style axiomatization we provide in \S~\ref{sec:logic} has certain advantages over those presented in earlier BI literature.



\subsection{Algebras}
The algebras of bunched implication logic are Heyting algebras with a residuated commutative monoid. 
 A \emph{Heyting algebra} $\m A = (A,\wedge,\vee,\to,\top,\bot)$ is a bounded lattice $(A,\wedge,\vee,\top,\bot)$ s.t. $\to$ is the residual of $\wedge$, i.e.,
$$
x\wedge y\le z\iff y\le x\to z\qquad\text{for all $x,y,z\in A$.}
$$
It follows from this property that $\wedge$ distributes over $\vee$, hence $(A, \wedge, \vee)$ is a distributive lattice, which in turn implies that $\vee$ distributes over $\wedge$.
In fact the residuation property easily implies the following stronger identities: if $\bigvee y_i$ and $\bigwedge y_i$ exist then
$$
x\wedge\bigvee y_i=\bigvee(x\wedge y_i), \ 
x\to\bigwedge y_i=\bigwedge(x\to y_i) \text{ and }
(\bigvee y_i)\to x=\bigwedge(y_i\to x).
$$
 Heyting algebras are the algebraic semantics of intuitionistic logic, with (non-classical) negation defined by $\neg x=x\to \bot$.

A \emph{generalized bunched implication algebra} (\emph{GBI-algebra}), is a tuple $(A,\wedge,\vee,\to,\top,\bot,\cdot,\backslash,/,1)$ where $(A,\wedge,\vee,\to,\top,\bot)$ is a Heyting algebra, $(A,\cdot,1)$ is a monoid and $\backslash,/$ are the left and right residuals of $\cdot$, i.e.,
$$
x\cdot y\le z\iff y\le x\backslash z \iff x\le z/y\qquad\text{for all $x,y,z\in A$.}
$$
We usually write $xy$ instead of $x\cdot y$, and assume that this operation has the highest priority, followed by $\backslash,/$, then $\wedge,\vee$, and finally $\to$.
The residuation property implies that for all existing meets and joins
$$
x(\bigvee y_i)=\bigvee xy_i\qquad
x\backslash\bigwedge y_i=\bigwedge x\backslash y_i\qquad
(\bigvee y_i)\backslash x=\bigwedge y_i\backslash x.
$$
$$
(\bigvee y_i)x=\bigvee y_ix\qquad
(\bigwedge y_i)/x=\bigwedge y_i/x\qquad
x/(\bigvee y_i)=\bigwedge x/y_i.
$$
 GBI-algebras have distributive residuated lattices as reducts. 
 Many other properties and results follow simply from this observation \cite{GalatosJKO07}.

\emph{GBI-homomorphisms} are functions that preserve all the operations, i.e., $f:A\to B$ such that $f(x\bullet y)=f(x)\bullet f(y)$ for all $x,y\in A$, $\bullet\in\{\wedge,\vee,\to,\cdot,\backslash,/\}$, $f(\bot)=\bot$ and $f(1)=1$ ($f(\top)=\top$ follows since $\bot\to\bot=\top$).

The classes of Heyting algebras and of GBI-algebras are denoted by $\algcl{HA}$ and $\algcl{GBI}$ respectively. They are both finitely based equational classes, meaning they are determined by finitely many equations (or inequations, since $s\le t$ and $s=s\wedge t$ are equivalent). For $\mathsf{HA}$ it suffices to take the equations of bounded lattices $(A,\wedge,\vee,\top,\bot)$ together with
$$
x\le y\to((x\wedge y) \vee z)\qquad
x\wedge(x\to y)\le y
$$
and for GBI-algebras one can add the inequations
$$
x\le (xy \vee z)/y\qquad 
((x/y)\wedge z)y\le x\qquad
x\le y\backslash(yx \vee z)\qquad
x((x\backslash y)\wedge z)\le y.
$$
By Birkhoff's \alo{HSP}-theorem, equational classes are precisely the classes that are varieties, i.e., closed under homomorphic images (\Ho), subalgebras (\So) and direct products (\Po), and moreover, for any class $\algvr K$ of algebras (of the same type) the variety $\Vo(\algvr K)$ generated by $\algvr K$ is $\alo{HSP}(\algvr K)$ (for details cf., e.g., \cite{Raftery:here}).

The variety of \emph{bunched implication algebras} (or BI-algebras) is the
 subclass $\algcl{BI}$ of all commutative GBI-algebras. In this case we use the 
more traditional notation of BI logic: $x*y=xy$ and $x\bi y=x\backslash y=y/x$. 
The  subvariety of \emph{Boolean BI-algebras}, axiomatized by $\neg\neg x=x$, is denoted by $\algcl{BBI}$.  

\subsection{Congruences}
An equivalence relation $\theta$ on an algebra $\m A$ is a \emph{congruence} if 
\begin{center}
$x\theta y$ implies $f(z_1,\ldots,x,\ldots,z_n)\theta f(z_1,\ldots,y,\ldots z_n)$
\end{center} for each argument of all fundamental operations $f$ of $\m A$. The set of all such congruences, ordered by inclusion, forms the congruence lattice Con$(\m A)$. The structure of this lattice determines several interesting properties of the algebra and of the class containing the algebras \cite{Raftery:here} so we now consider how to determine congruences of GBI-algebras.

An algebra with a constant operation $e$ is called $e$-\emph{congruence regular} if every congruence relation $\theta$ is determined by its $e$-congruence class 
$[e]_\theta=\{x:x\theta e\}$, i.e.,
$$\text{for all $\theta,\psi$, }[e]_\theta=[e]_\psi\implies \theta=\psi.$$ 
In such an algebra it suffices to describe the $e$-congruence classes,
and the poset of these classes, ordered by inclusion, is isomorphic to the
congruence lattice of the algebra.

For example, groups are $1$-con\-gruence  regular and
Heyting algebras are $\top$-con\-gruence regular, whereas monoids and lattices
(even with bounds and distributivity) are not $e$-congruence regular with respect to any constant operation $e$. 
In the case of Heyting algebras, the $\top$-congruence classes are precisely the lattice-filters of the algebra, i.e., sets $F\subseteq A$ such that ${\uparrow} F\subseteq F$ and $x,y\in F$ imply $x\wedge y\in F$. The congruence $\theta_F$ associated with the filter $F$ is defined by $x\theta_F y\iff x\to y,\ y\to x\in F$.

If we ignore the Heyting operations $\to, \top, \bot$ then GBI-algebras are residuated lattices. Congruence classes of lattices are always convex (i.e., if $x\le y\le z$ and $x,z$ are in a class, then $y$ is also in the class), and as mentioned above, residuated lattices are $1$-congruence regular (see e.g. \cite{JipsenTsinakis02}). The $1$-congruence classes are precisely the convex subalgebras $\m C$ that are closed under conjugation, i.e., for all $a\in A$ and all $x\in C$ it follows that $a\backslash xa\wedge 1, ax/a\wedge 1\in C$. Given such a congruence class $C$, the congruence $\theta_C$ is determined by $x\theta_Cy\iff x\backslash y\wedge y\backslash x\wedge 1\in C$.

From these observations we conclude that GBI-algebras are both $1$-con\-gruence regular and $\top$-congruence regular. 
The following criterion can be used to check whether a residuated lattice congruence is a GBI-congruence.

\begin{theorem}
Suppose $\theta$ is a residuated lattice congruence on the $\to,\bot,\top$-free
reduct of a GBI-algebra $\m A$. Then the following are equivalent: \begin{enumerate}
\item $\theta$ is a GBI-congruence.
\item For all $x,y,z\in A$ if $x\theta y$ then $x{\to}z\,\theta\,y{\to}z$ and $z{\to}x\,\theta\,z{\to}y$.
\end{enumerate}
\end{theorem}


\subsection{Logic} \label{sec:logic}
Throughout we use algebraic term syntax for logical formulas. Propositional intuitionistic logic uses the symbols $\wedge,\vee,\to,\top,\bot$ to build formulas (= terms) from variables $x,y,z,x_1,\ldots$, and $x\leftrightarrow y$ abbreviates the formula $(x\to y)\wedge(y\to x)$. The consequence relation $\vdash_\text{IL}$ of intuitionistic logic is defined by the following Hilbert system, traditionally denoted by HJ. The axioms are all substitution instances of the formulas below, and modus ponens is the only inference rule:
$$\inferrule{x\quad x\to y}{y}\qquad
\bot\to x\qquad
x\to\top
$$
$$
x\to(y\to x)\qquad
(x\to(y\to z))\to((x\to y)\to(x\to z))\qquad
$$
$$
x\wedge y\to x\qquad
x\wedge y\to y\qquad
x\to(y\to x\wedge y)
$$
$$
x\to x\vee y\qquad
x\to y\vee x\qquad
(x\to z)\to((y\to z)\to(x\vee y\to z))
$$
Given a set $\Gamma$ of formulas, $\Gamma\vdash_\text{IL}\varphi$ holds if there is a finite sequence of formulas $\varphi_1,\ldots,\varphi_n=\varphi$ such that each $\varphi_i$ is in $\Gamma$, is an axiom, or is the result of modus ponens applied to $\varphi_j,\varphi_k$ for some $j,k<i$. For example, a standard (but rather non-obvious) deduction shows that $\to$ is transitive: $\{x\to y,y\to z\}\vdash_\text{IL}x\to z$. The theorems (tautologies) of IL are all the formulas $\varphi$ such that $\emptyset\vdash_\text{IL}\varphi$, or equivalently $\vdash_\text{IL}\varphi\leftrightarrow\top$.

HJ is extended to a Hilbert system HGBI for $\algcl{GBI}$ by adding symbols $\cdot,\backslash,/$, all substitution instances of the following formulas as axioms 
$$
(xy)z\leftrightarrow x(yz)\qquad 
1x\leftrightarrow x\qquad
x1\leftrightarrow x\qquad
$$
and the bidirectional residuation rules
$$
\mprset{fraction={===}}
\inferrule{xy\to z}{y\to x\backslash z}\qquad
\inferrule{xy\to z}{x\to z/y}.
$$
While many other axioms (or rules) can be used to axiomatize $\algcl{GBI}$, 
 the axiomatization given here emphasizes the close relationship between Hilbert systems (of sufficient strength) and equational deduction. Neither approach is as effective as the sequent calculus decision procedure that is outlined in \S~\ref{ProofTheory}, but Birkhoff's system of equational deduction requires perhaps less explanation than the corresponding logical systems, and it allows \emph{substitution of equal terms} based on equalities derived from the axioms or assumptions.
\begin{theorem} \label{th:logalg}
HGBI corresponds to \GBI: $\Gamma\vdash_\text{HBGI}\varphi$ if and only if $\varphi=\top$ can be derived by equational reasoning from $\{\gamma\leftrightarrow\top:\gamma\in \Gamma\}$.
\end{theorem}
\begin{proof}
Analogous results are well known for intuitionistic and substructural reducts of GBI; cf. Remark \ref{rem:alg}. The HGBI axioms clearly correspond to the equational monoid axioms. It remains to show that Birkhoff's congruence rules are derivable in HGBI. For example the rule $\inferrule{x\to y}{y\backslash z\to x\backslash z}$ is proved as follows:
$$
1.\ x\to y\quad 
2.\ y\backslash z\to y\backslash z\quad
3.\ y(y\backslash z)\to z\quad
4.\ y\to z/(y\backslash z)
$$
$$
5.\ x\to z/(y\backslash z)\quad 
6.\ x(y\backslash z)\to z\quad
7.\ y\backslash z\to x\backslash z
$$
where transitivity of $\to$ was used for step 5. 
 \takeout{Hence $\inferrule{x\leftrightarrow y}{y\backslash z\leftrightarrow x\backslash z}$.}
 The rules $\inferrule{x\leftrightarrow y}{z\backslash x\leftrightarrow z\backslash y}$,
$\inferrule{x\leftrightarrow y}{xz\leftrightarrow yz}$,
$\inferrule{x\leftrightarrow y}{zx\leftrightarrow zy}$,
$\inferrule{x\leftrightarrow y}{x/z\leftrightarrow y/z}$,
$\inferrule{x\leftrightarrow y}{z/y\leftrightarrow z/x}$ are proved similarly.
\hfill\qed\end{proof}
A Hilbert system HBI for bunched implication logic is obtained by adding an axiom $x*y\to y*x$, in which case the rules for $/$ can be omitted, and the rules for $\backslash$ are rewritten with $\bi$. The resulting system is similar to the one in \cite{Pym02:book}. Alternatively one can use the following system with two more axioms and simpler rules (in addition to the axioms and rules of HJ).
$$
(x*y)*z\leftrightarrow x*(y*z)\qquad 
x*y\to y*x\qquad
x*1\leftrightarrow x\qquad
x*(x\bi y)\to y
$$
$$
x\bi (y\bi z)\leftrightarrow x*y\bi z\qquad 
\inferrule{x\to y}{x*z\to y*z}\qquad
\inferrule{x\to y}{1\to x\bi y}.
$$

\begin{rem} \label{rem:alg}
The connection between HGBI (HBI) and \GBI\ (\BI) exposed by Theorem \ref{th:logalg} is an instance of the phenomenon known as \emph{algebraizability} \cite{BlokP89:ams}. In this volume, an overview is provided by Font \cite{Font:here}. For Heyting and substructural reducts of (G)BI, details can be found, e.g., in Galatos et al. \cite[\S\S\ 1.4.3, 2.6]{GalatosJKO07}. In fact, due to the presence of Heyting $\to$, logics extending HGBI belong to a particularly well-behaved class of \emph{Rasiowa implicative} logics \cite{CintulaN10,Rasiowa74} \cite[\S\ 5]{Font:here}.
\end{rem}

\section{Concrete Models: Standard Models of BI} \label{ModelsBI}

As we already suggested in \S~\ref{sec:intro}, GBI admits a wealth of practically motivated models. 
 In this section, we focus on  models  previously investigated in the commutative setting, using  Pym et al. \cite[\S~4]{PymOHY04:tcs} as our blueprint. Even here, we are going to see potential for non-commutative generalizations (cf. especially \S~\ref{sec:semi}); we are going to explore more ``natively non-commutative'' models in \S~\ref{ModelsGBI}.

\subsection{Generalized PPMs}

In order to streamline the discussion and facilitate checking the GBI axioms for a large class of examples, let us follow the example of Pym et al. \cite[\S~3.6]{PymOHY04:tcs} (see also \cite{GalmicheMP05:mscs}) and define a convenient semantics
\begin{itemize}
\item whose defining properties are easily verifiable,
\item which covers many natural models and yet
\item  avoids the full generality of semantics for distributive substructural logics based on ternary relations. 
\end{itemize}
More specifically, consider preordered partial monoids 
 $\Ppm{M} = (M, \ppm, E, \ppmsb)$ with $\cdot$ lifted to an operation on subsets by $X \cdot Y = \{ z \mid \exists x \in X, y \in Y. x \cdot y \ppmsb z \}$  for $X, Y \subseteq M$ and with
 
 \begin{itemize}
 \item $(M,\ppm)$ being a partial semigroup up to the equivalence relation $\equiv$ defined as $\ppmsb \cap \ppmsp$, i.e., whenever one element of $\{ x \cdot (y \cdot z), (x \cdot y) \cdot z\}$ exists, the other one exists as well, and they are in the same $\equiv$-class,
 \item $E \subseteq M$ being a collection of \emph{unit elements}, i.e., 
 for any $x \in M$, $$\emptyset \neq \{x\} \cdot E  \; \subseteq \; [x]_\equiv \qquad \emptyset \neq  E \cdot \{x\} \; \subseteq \; [x]_\equiv$$
 (we always assume the closure of $E$ under $\equiv$),
 \item the \emph{bifunctoriality condition} 
 \begin{center}
 $x \ppmsb x'$, $y \ppmsb y'$ and $x' \ppm y' \in M$ implies $x\ppm y\in M$ and $x \ppm y \ppmsb x' \ppm y'$
 \end{center}
 holding for any $x, x', y, y' \in M$.
 \end{itemize}
Such a structure will be called a \emph{generalized PPM} (short for \emph{preordered partial monoid}), and the notation $x\cdot y\in M$ is used to indicate that $x\cdot y$ is defined. Most of the time, we will restrict attention to the case where $E = \{\ppmu\}$ for some $\ppmu \in M$; in such a case, we will speak of a \emph{\bro proper\brc PPM}. Define now $\Cmx{\Ppm{M}}$, \emph{the complex algebra of} $\Ppm{M}$,  as the algebra in GBI-signature whose
\begin{itemize} 
\item universe consists of all upsets of $(M,\ppmsb)$,
\item Heyting connectives are interpreted in the standard intuitionistic way,
\item the unit element is defined as the upset of $E$,
\item $X \cdot Y$ is as defined above,
\item the residuals are obtained using the fact that upsets are closed under arbitrary unions and $\cdot$ distributes over these unions.
\end{itemize}

\begin{fact} \label{fact:ppm}
The complex algebra of any PPM is a GBI-algebra.
\end{fact}

\noindent
This is a straightforward generalization of the facts used by Pym et al. \cite{PymOHY04:tcs} and Galmiche et al.  \cite{GalmicheMP05:mscs} with the obvious difference that we are not assuming commutativity. Of course, from the point of view of a substructural logician, such partial monoids can be turned into instances of ternary relation semantics by setting $Rxyz$ whenever $x \cdot y = z$. We will return to  relational semantics in \S~\ref{DualityGBI} and \S~\ref{ProofTheory}, each time with a somewhat different focus and somewhat different notation.

\subsection{Intuitionistic vs. Classical Resource Models}

Most PPM-style models discussed below can be in fact obtained in two flavours: a monoid with a degenerate or discrete order (thus yielding a boolean GBI-algebra) and an associated intuitionistic structure with a nontrivial order definable in terms of the monoid operation. This has been noted early on in the development of BI logic, leading to G\"odel-McKinsey-Tarski-style modal translations between intuitionistic and classical logics of suitable classes of models \cite[Prop. 9]{IshtiaqOH01:popl} (see also \cite{Galmiche2006}) and commonly used terms \emph{intuitionistic semantics} and \emph{intuitionistic assertion}.\footnote{Interestingly, 
 the intuitionistic BI is embeddable in the classical BBI \cite{Larchey-WendlingG09}, rather than the other way around. Indeed, as pointed out in Litak et al. \cite[\S~4.1]{LitakPR17}, (un)decidability results discussed in \S~\ref{Decidability} entail that no negative translation from BBI to BI can work.} 

Thus, the idea is well known, but in our setting we can present it in a rather convenient way. Let a 
 PME (a \emph{partial monoid up to equivalence}) 
 be a generalized PPM where the ordering is an equivalence relation. Obviously, this happens  iff $\ppmsb$ coincides with $\equiv$ as defined above. 
 
 We speak of CPME (\emph{commutative partial monoid up to equivalence}) when the commutativity law $x \cdot y \equiv y \cdot x$  holds.  While commutativity makes transition from the boolean to the intuitionistic setting much smoother, we can do without it, at the expense of introducing some additional apparatus. 
 Given  a (not necessarily commutative!) PME $\Ppm{M} = (M, \ppm, E, \equiv)$, let  $\CM$ be 
   the collection of those  $x \in M$ for which
 \begin{itemize}
 \item for any $y \in M$, whenever $x \cdot y \in M$, there exists $y'$ s.t. $x \cdot y  \equiv  y' \cdot x$ and
 \item for any $y \in M$, whenever $y \cdot  x \in M$, there exists $y'$ s.t. $y \cdot x \equiv  x \cdot y'$.
 \end{itemize}

 
 \begin{fact} \label{fact:cm}
In any PME $\Ppm{M} = (M, \ppm, E, \equiv)$:
 \begin{upbroman}
 \item $E \subseteq \CM$,
 \item 
 $\CM \cdot \CM \subseteq \CM$, that is, $\CM \subseteq \CM \backslash \CM$, 
 \item  For any $X \in \Cmx{\Ppm{M}}$, $X \cdot \CM = \CM \cdot X$ and $\CM \backslash X = X / \CM$,
 \item Whenever $\Ppm{M}$ is commutative, $\CM = M$.
\end{upbroman}
\end{fact}
 
 \noindent
 Now let us define the \emph{substate relation}\footnote{In the theory of semigroups, one would rather use the name \emph{algebraic preordering}. It is also known as \emph{divisibility relation} for commutative semigroups.}, generalizing the corresponding definition for \emph{separation algebras} (see below) proposed by Calcagno et al. \cite{Calcagno2007}:
\begin{center}
$x \ipmsb y$ \qquad iff \qquad $\exists z \in \CM.\, x \cdot z \equiv y$.
\end{center}

\begin{theorem} \label{inttocl}
Let $\Ppm{M} = (M, \ppm, E, \equiv)$ be a PME. Then
 \begin{upbroman}
\item the complex algebra $\Cmx{\Ppm{M}}$ of unions of equivalence classes is a boolean GBI-algebra,
\item $\ipmsb$ is a preorder, with the associated equivalence relation ${\intu{\equiv}} = \ipmsb \cap \ipco$ containing the original $\equiv$,
\item $\intu{\Ppm{M}} = (M, \ppm, E, \ipmsb)$ is a generalized PPM and hence $\Cmx{(\intu{\Ppm{M}})}$ is a GBI-algebra, 
\item Elements of $\Cmx{(\intu{\Ppm{M}})}$ are exactly those sets $A$ of equivalence classes in $\Ppm{M}$ which in $\Cmx{\Ppm{M}}$ satisfy one of the following equivalent conditions:

\medskip

\begin{tabular}{llll}
$A \subseteq A/\CM$, & $\CM \subseteq A \backslash A$, & $A \cdot \CM \subseteq A$, & $A \cdot \CM = A$, \\
$A \subseteq \CM \backslash A$, & $\CM \subseteq A / A$, &  $\CM \cdot A \subseteq A$, &  $\CM \cdot A = A$.
\end{tabular}
\medskip
\item[\bro v\brc] For any $A, B \in \Cmx{(\intu{\Ppm{M}})}$ s.t. $B \subseteq \CM / \CM$ \bro where $\CM / \CM$  denotes an element of  $\Cmx{\Ppm{M}}$ rather than $ \Cmx{(\intu{\Ppm{M}})}$\brc,  $A \cdot B \leq A$ holds in $\Cmx{(\intu{\Ppm{M}})}$.  Hence, whenever $\Ppm{M}$ is a CPME, $\Cmx{(\intu{\Ppm{M}})} \in \mathsf{BI_w}$ \bro cf. \S~\ref{Subvarieties}\brc.
\end{upbroman}
\end{theorem}

\begin{proof}\

 \begin{upbroman}
\item A direct corollary of Fact \ref{fact:ppm}: unions of equivalence classes are upsets of PME and they are closed under complementation. 

\item Transitivity of $\ipmsb$  follows from associativity of $\cdot$ and Fact \ref{fact:cm}.(ii), whereas  reflexivity of $\ipmsb$  follows from the monoidal unit law and Fact \ref{fact:cm}.(i) (both up to bifunctoriality of $\equiv$ and equivalence). The latter assures also the containment claim, jointly with transitivity of $\equiv$.  

\item In the light of (ii), we only need to ensure bifunctoriality of $\ipmsb$; note that this is the first time when we use the fact that the substate relation is defined in terms of $\CM$. Assume $z_x, z_y \in \CM$ s.t. $x \cdot z_x \equiv x'$ and $y \cdot z_y \equiv y'$. By bifunctoriality of $\equiv$, we get that $(x \cdot z_x) \cdot (y \cdot z_y) \equiv x' \cdot y'$. Iterating the associativity law yields $x \cdot (z_x \cdot y) \cdot z_y \equiv x' \cdot y'$. Now we use the definition of $\CM$ to pick a suitable $z'_x$ s.t.  $x \cdot y  \cdot z'_x \cdot z_y \equiv x' \cdot y'$. Thanks to Fact  \ref{fact:cm}.(ii), we obtain that $ z'_x \cdot z_y \in \CM$.

\item By definition, a set of equivalence classes $A$ is an element of $\Cmx{(\intu{\Ppm{M}})}$ iff for any $a \in A$ and $c \in \CM$,  it is the case that $a \cdot c \in A$. This is an equivalent way of stating that $\Cmx{\Ppm{M}}$, it holds that $A \cdot \CM \subseteq A$. The rest follows form Fact \ref{fact:cm}.(iii).

\item 
Assume $A, B \in \Cmx{(\intu{\Ppm{M}})}$, $B \subseteq \CM$ and $x \in A \cdot B$. That is, there are $a \in A, b \in B, z \in \CM$ s.t. $a \cdot b \cdot z \equiv x$. 
 By assumption, we get that $b \cdot z \in \CM$, thus $a \ipmsb x$ and we just use the fact that $A$, like all elements of $\Cmx{(\intu{\Ppm{M}})}$, is $\ipmsb$-upward closed. \hfill\qed
\end{upbroman}
 
 \end{proof}

One problem with this construction is that ${\intu{\equiv}}$ may happen to be bigger than the original $\equiv$, also when $\equiv$ is just the diagonal (equality relation). 
 Let us say that a PME \begin{itemize}
 \item  is \emph{right-cancellative} if $x \cdot y \equiv x \cdot y'$  implies $y \equiv y'$  (in the presence of commutativity this implies left-cancellativity) and
\item satisfies \emph{indivisibility of units} if
$x \cdot y \in E$ implies $x \in E$ (and hence also $y \in E$).
 \end{itemize}


\begin{theorem} \label{th:cancel}
Let $\Ppm{M} = (M, \ppm, E, \equiv)$ be a PME. Whenever  $\Ppm{M}$ is right-cancel\-lative and satisfies indivisibility of units, then the associated equivalence relation ${\intu{\equiv}} = \ipmsb \cap \ipco$ 
 of  $\intu{\Ppm{M}}$ is the same as the original $\equiv$. 
 \end{theorem}

\begin{proof}
We have already established in Theorem \ref{inttocl}.(ii) that ${\intu{\equiv}}$ contains the original $\equiv$. Thus, we have only to show the converse inclusion. Assume then that $z_x, z_y \in \CM$, $x \cdot z_x \equiv y$ and $y \cdot z_y \equiv x$. Therefore,  $x \cdot z_x \cdot z_y \equiv x \cdot e$.   Right-cancellativity implies that $z_x \cdot z_y \equiv e$ and indivisibility of units implies that $z_x \in E$, hence $x \equiv y$.     
\hfill\qed\end{proof}
Whenever  $\equiv$ in $\Ppm{M}$ is the equality relation (as it happens in most natural examples), Theorem \ref{th:cancel} says that cancellativity and indivisibility of units of $\Ppm{M}$ entail that $\ipmsb$ is a partial ordering.

\emph{Separation algebras} \cite{Calcagno2007}  are  CPME's (i.e., \emph{commutative} PME's) which moreover are cancellative,  have $E=\{e\}$, i.e., are  \emph{\bro proper\brc PPM}s in our terminology, and where $\equiv$ is just the identity relation. 

We are now ready to instantiate this framework to specific applications.

\subsection{Resource Allocation and Generalized Effect Algebras} \label{Effect}

Given any set (which is thought of as the supply of \emph{resources}), one can impose a separation algebra structure on the set of all its subsets by taking $x \ppm y$ to be $x \cup y$ whenever these two sets are disjoint and undefined otherwise. The order can be taken to be discrete (equality) or one can transfer it via Theorems \ref{inttocl} and \ref{th:cancel}, obtaining ordinary inclusion relation as the ordering. 
The empty set is the identity element, and the collection of finite subsets forms a subalgebra. This example is discussed in detail in \S~4.3 of Pym et al. \cite{PymOHY04:tcs} (see also \cite{IshtiaqOH01:popl,Reynolds00:intuitionistic,Reynolds02:lics}). 

A \emph{generalized effect algebra} is a separation algebra that satisfies the \emph{positivity law}: if $x\ppm y=e$ then $x=e=y$. 
 This holds, e.g., for the separation algebra defined by disjoint union. 
 The more specialized concept of \emph{effect algebra} was defined by Foulis and Bennett \cite{FoulisBennett94} as an abstraction of \emph{quantum effect operators} in Hilbert space (i.e., self-adjoint operators with spectrum in the unit interval). These are generalized effect algebras with a constant $\top$ satisfying in addition the \emph{orthosupplementation} law: for every $x$, there exists a unique $y$ s.t. $x \cdot y$ exists and $x \cdot y = \top$.

\subsection{Resource Separation, Memory and the Heap Model}   \label{HeapModel}

The next example follows the same idea as the separation algebra given by
disjoint union
\cite{IshtiaqOH01:popl,Reynolds00:intuitionistic,Reynolds02:lics,PymOHY04:tcs,OHearn12:nato,DemriD15:jancl}. It is also probably the one most responsible for the success of BI in computer science. This time, resources are interpreted concretely  as portions of computer memory. A good overview of various possible notions of \emph{memory models} can be found in the recent work of Brotherston and Kanovich \cite[\S~2]{BrotherstonK14:jacm} and also in Demri and Deters \cite{DemriD15:jancl}.

\newcommand{\nil}{\mathsf{nil}}

More specifically, given an infinite set of \emph{locations} $L$ and a set of \emph{record values} $RV$, the latter possibly with some additional structure, we define \emph{heaps} (or \emph{heaplets}, as suggestively named by Berdine et al. \cite{Berdine2006}) as finite partial functions from $L$ to $RV$. Particularly when reasoning about linked data structure, it is  common to demand 
 that we have in addition a function from $RV$ to $L~\cup~\{\nil\}$, 
 where $\nil$ is a fixed null pointer. Actually,  separation logic overviews quite often restrict attention to single-linked lists, defining $RV$ to be $V \times (L \cup \{ \nil \})$ and the set of \emph{base values} $V$ is typically taken to be, e.g., $\inte \cup L \cup \{ \nil \})$. 
There are other possible choices for $RV$, for example it can be taken to be $V^2$. 

One obtains a separation algebra structure on heaps by setting $h \ppm h'$ to be their union when domains of $h$ and $h'$ are disjoint and undefined otherwise. Again, 
  the intuitionistic option offered by Theorems \ref{inttocl} and \ref{th:cancel} 
 orders heaps by inclusion between their graphs. Interestingly, one of the earliest papers 
 on separation logic \cite{Reynolds00:intuitionistic} took the latter route (cf.  also \cite{IshtiaqOH01:popl,Reynolds02:lics,PymOHY04:tcs}). 

Let us mention here one more possible tweak to these models, which makes them  closer to memory models of actual programming languages and more convenient from the point of view of development of program logics as discussed in \S~\ref{SeparationLogic}. It is also our first opportunity to use generalized rather than proper PPM's. Namely, assume that in addition to the collection of locations $L$, we also have a collection of ordinary \emph{program variables} $\mathit{PVar}$, and in addition to record values $RV$, we also have \emph{store values} (or \emph{stack values}) $\mathit{Val}$. We define then \emph{stores} (or \emph{stacks}) as mappings from $\mathit{PVar}$ to $\mathit{Val}$, either total or (finite) partial ones,\footnote{Brotherston and Kanovich \cite[\S~2.2]{BrotherstonK14:jacm} stick to the finite partial definition, but the total one is arguably more natural and common (see e.g. \cite{PierceSF,Winskel93,DemriD15:jancl}); this is one of differences between stores (stacks) and heap(let)s. Especially under the total perspective the name \emph{store}, used also by Demri and Deters \cite{DemriD15:jancl} seems more adequate. Nevertheless, both perspectives can be brought together: one can think of the constant function $\lambda x : \mathit{PVar}. 0$ as the \emph{default} or \emph{unintialized} stack (store) and restrict the attention to stacks almost everywhere equal to zero.} and  the \emph{store-and-heap}  (or \emph{stack-and-heap})  \emph{model} \cite[\S~2.2]{BrotherstonK14:jacm}) as consisting of pairs $(s, h)$ with $s$ a store and $h$ a heap. The set of units $E$ is defined then as the collection of all pairs $(s,\emptyset)$.  We say $(s, h) \cdot (s', h')$ is defined whenever $s = s'$ and $h \cdot h'$ (as introduced above) is defined. As we are going to see in \S~\ref{SeparationLogic}, having stores at our disposal we do not need anymore the above-mentioned restrictions on the structure of $RV$ such as the one that each $RV$ should contain a pointer to another $RV$.

\pjnt{Which additional conditions does the preceding sentence refer to?}\tlnt{is it clear now?}

\subsection{Ambient Logic, Trees and Semistructured Data} \label{sec:semi}

Pym et al. \cite[\S~4.2]{PymOHY04:tcs} illustrate how to obtain a PPM using  Cardelli and Gordon's \emph{ambient logic} \cite{CardelliG00:popl}. This influential formalism was developed  further in a number of references, some of them  focusing on reasoning about trees and semistructured data \cite{CardelliG04:mscs}.  Subsequent developments included \emph{context logic}, a formalism specifically intended for analyzing dynamic
updates to tree-like structures with pointers (such as XML
with identifiers and idrefs) \cite{CalcagnoGZ05:popl,CalcagnoDY10:ic}. As this example generalizes particularly nicely to the non-commutative setting, we discuss it in more detail. 

Consider a set of labels $Lab$. The set of \emph{labelled trees} (which might be more adequately called \emph{labelled forests} in the terminology of W3C specifications 
  \cite[\S~3.4]{CardelliG04:mscs}) is given by the following syntax:
\newcommand{\fore}{\cdot}
$$
S, T ::= 0 \mid a[S] \mid S \fore T,
$$
where $a \in Lab$. This is a standard way to represent semistructured data like XML documents. One identifies forests using the equivalence relation generated by associativity of $\fore$ and $0$ being a neutral element for $\fore$, which obviously yields a generalized PPM.  
 In this free construction of labelled trees, 
the operation $\fore$  is total and indeed it was intended to be total in several references \cite{CardelliG04:mscs,PymOHY04:tcs}, but as pointed out by e.g. Calcagno et al. \cite[\S~2]{CalcagnoGZ05:popl}, it is natural to restrict the attention to trees with uniquely identifying labels. Under such an assumption, $S \fore T$ is defined only if the labels occurring in $S \fore T$ are disjoint---and thus we have yet another example of a partially defined monoid. 

However, from our point of view it is even more interesting to note that while almost all references mentioned in this subsection insists on commutativity of $\fore$, it is hardly the most obvious  assumption. In fact, not only are XML documents defined as finite \textit{sibling-ordered} trees, but official specifications of languages standardized by W3C for the purpose of querying and navigating XML documents like XPath and XQuery allows explicit access to the sibling order (see, e.g., ten Cate et al. \cite{tCateLM10:jal,tCateFL10:jancl} for more information, 
 including a discussion of the relationship  of these formalisms to modal logics). And, needless to say, any representation of trees for storage or manipulation purposes would  involve ordering on nodes; in short, $\fore$ should be thought of as creating lists rather than multisets. While this issue is occasionally discussed in the 
 literature 
 \cite[\S\S~3.1 \& 3.4]{CardelliG04:mscs}, most references 
 tend to glide over this problem. Dropping the requirement of commutativity makes the complex algebra of such a PPM an instance of a GBI algebra which is not a BI algebra.

Finally, let us note that one obtains 
 a GBI algebra with a non-boolean Heyting reduct by replacing the discrete order on trees by $S \ppmsb T$ defined as ``$S$ is a \emph{generated subtree} (or, strictly speaking, a \emph{generated subforest}) of $T$''.

\subsection{Costs, Logic Programming and Petri Nets} \label{sec:petri}

Pym et al. \cite[\S~4]{PymOHY04:tcs} describe three other classes of  CS-motivated PPM's giving rise to natural BI complex algebras. In brief, they are  as follows:
\begin{itemize}
\item an adjustment of the Petri net semantics of linear logic described by Engberg and Winskel \cite{EngbergW97:apal}. An interesting feature of this example is that the PPM in question illustrates the benefits of allowing preorders instead of insisting on posets. Modelling of  Petri nets using separation algebras is discussed by Calcagno et al. \cite[\S~2]{Calcagno2007};
\item a logic programming model of Armel{\'{\i}}n  and Pym \cite{ArmelinP01:ijcar}  based on a commutative total PPM of \emph{hereditary Harrop bunches} and 
\item a \emph{money and cost} example, 
  tailored to highlight both similarities and differences with Girard's \emph{Marlboros and Camels} linear logic example.
\end{itemize}
 Brotherston and Calcagno \cite[\S~5]{BrotherstonC10:lmcs} provide some additional commutative models, focusing on \emph{involutive} boolean ones, i.e., those whose dual algebras belong to the variety denoted in \S~\ref{Subvarieties} as $\mathsf{InBBI}$ (Brotherston and Calcagno \cite{BrotherstonC10:lmcs} use the term \emph{classical}).


\section{Essentially Noncommutative Models} \label{ModelsGBI}

Finally, we present two more classes of examples, illustrating the advantages of dropping the assumption of commutativity even more starkly than \S~\ref{sec:semi}.

\subsection{Weakening Relations and Relation Algebras} \label{sec:weakening}
\newcommand{\inv}[1]{{\Ppm{#1}}^{\partial}}
\newcommand{\relc}{\,\circ\,}
\newcommand{\wrel}[1]{\mathit{wRel}\Ppm{#1}}

This example involves  generalized 
 PPM's. Consider a poset $\Ppm{X}=(X, \ppmsb)$. Say that $R$ is a \emph{weakening relation} on $\Ppm{X}$ iff ${\ppmsb} \relc R \relc {\ppmsb} = R$, where $\relc$ is the relation composition. The collection of all the weakening relations on $\Ppm{X}$ is written as $\wrel{X}$.

\begin{fact} \label{fact:wrel}
$\wrel{X}$ is closed under arbitrary unions $\bigcup$ and intersections $\bigcap$, with $\relc$ distributing over $\bigcup$ and $\ppmsb$ being the neutral element of $\relc$. 
\end{fact}

Consequently,  $\wrel{X}$ carries the structure of a GBI algebra and is called the \emph{full weakening relation algebra} generated by $\Ppm{X}$.

It is possible to see $\wrel{X}$ as a complex algebra of a  (generalized!) PPM. Set $\inv{X} = (X, \ppmsp)$ and consider 
$M = \Ppm{X} \times \inv{X}$. The set of unit elements $E$ and the partial monoid operation $\cdot$ are defined then in an obvious way, i.e., $E = \{ (x,x) \mid x \in X \}$ and $(x, y) \cdot (y', z)=(x, z)$ whenever $y = y'$ and undefined otherwise. Fact \ref{fact:wrel} is then obtained as a corollary of Fact \ref{fact:ppm}.


What is particularly interesting about this example is that when we restrict attention to discrete $\ppmsb$, we obtain exactly what is known as \emph{full set relation algebras} \cite[Def. 5.3.2]{HenkinMT85:book2} or \emph{square relation algebras} \cite[Ch. 6.0.3]{Maddux2006:book}. 

\subsection{Language Models} \label{sec:language}

Consider an alphabet $\Sigma$; as usual, we write the set of words in $\Sigma$ as $\Sigma^*$. The notions of \emph{language} and \emph{regular language} are standard and so is the notion of composition of languages. It is well known  that the set of regular languages, just like of all languages, is closed under finite unions and intersections, residuals and boolean complementation (cf., e.g., \cite[\S~3.2]{Pratt91:lncs}). Therefore, both arbitrary languages and regular languages over a given $\Sigma$ form a nice example of a boolean GBI. In fact, we can see this as another instance of the PPM setting, but once again dropping the assumption of commutativity was crucial to achieve full generality.


\section{Subvarieties of GBI-algebras and InGBI-algebras} \label{Subvarieties}

 A subvariety of \GBI\ (the variety of all GBI-algebras) is any subclass that is closed under \alo{HSP}, or equivalently any subclass that is determined by a set of identities (including the equational axioms of \GBI). The collection of all subvarieties of \GBI\ is denoted by $\Lambda_\text{GBI}$ or simply $\Lambda$. Since subvarieties are determined by sets of identities, $\Lambda$ contains at most continuum many subvarieties. Jankov \cite{Jankov68} showed that this upper bound is reached by 
 subvarieties of Heyting algebras, hence the same is true for GBI. Subvarieties are ordered by inclusion, and $\Lambda$ is in fact an algebraic distributive lattice, with $\algvr U\wedge \algvr W=\algvr U\cap \algvr W$ and $\algvr U\vee \algvr W=\Vo(\algvr U\cup\algvr W)$. The least element is the trivial variety $\algcl O$ of one-element GBI-algebras, and the largest element is $\algcl{GBI}$. 

For an involutive GBI-algebra, we first need to expand the language with a new constant symbol $0$, which is used to term-define the \emph{linear negations} ${\sim} x=x\backslash 0$ and $-x=0/x$. Then we add the identities 
${\sim}{-}x=x=-{\sim}x$ to define the variety \algcl{InGBI}.

Some prominent subvarieties\footnote{The reader is encouraged to compare this list with Galatos et al. 
 \cite[\S~3.5]{GalatosJKO07}.} of \algcl{GBI} and \algcl{InGBI} are:

\begin{itemize}
\item The variety $\algcl{BI}$ defined relative to $\algcl{GBI}$ by $xy=yx$.
\item The variety $\algcl{GBI_w}$ of GBI-algebras that satisfy the structural rule of \emph{weakening}, defined by the identity $x\cdot y\le x$, or equivalently by $\top =1$.
\item The variety $\algcl{BGBI}$ of \emph{Boolean GBI-algebras}, defined by $\neg\neg x=x$. It is also known as $\algcl{RM}$, the variety of \emph{residuated Boolean monoids} \cite{Jipsen92}.
\item The variety \algcl{CyGBI} of cyclic involutive GBI algebras, defined relative to \algcl{InGBI} by ${\sim} x=-x$.
\item The variety \algcl{InBI} of involutive BI-algebras, defined relative to \algcl{InGBI} by $xy=yx$.
\item The variety $\algcl{wRRA}$ of \emph{weakening representable relation algebras}, generated by all full weakening relation algebras (cf. \S~\ref{sec:weakening}). 
\item The variety $\algcl{SeA}$ of \emph{sequential algebras}, defined relative to $\algcl{BGBI}$ by the Euclidean law 
\begin{center}
$(x\triangleright y)\cdot z\le x\triangleright(y\cdot z)$ where $x\triangleright y=\neg(x\backslash\neg y)$ \cite{JipsenMaddux97}.
\end{center}
\item The variety $\algcl{RA}$ of \emph{relation algebras}, defined by $x\triangleright y=(x\triangleright 1)y$ \cite{JonssonTsinakis93}. 
The term $x\triangleright 1$ is
the \emph{converse operation} in relation algebras, denoted by $x^\smallsmile$.
\item The variety $\algcl{RRA}$ of \emph{representable relation algebras}, generated by all full relation algebras  (cf. \S~\ref{sec:weakening}). 
\item The variety $\algcl{CRA}$ of \emph{commutative relation algebras}, defined relative to $\algcl{RA}$ by $xy=yx$.
\item The variety $\algcl{GRA}$ of group relation algebras, generated by all complex algebras of groups.
\item The variety $\algcl{SRA}$ of \emph{symmetric relation algebras}, defined relative to $\algcl{RA}$ by $x^\smallsmile =x$.
\item The variety $\algcl{BBI}$ of \emph{Boolean BI-algebras} (= $\algcl{CRM}$ in \cite{Jipsen92}).
\item $\algcl{LGBI}$ is generated by all \emph{linearly ordered GBI-algebras}, or equivalently defined by the identity
\begin{center}
 $(x\to y)\vee(y\to x)=\top$.
 \end{center}
\item The variety $\algcl{BLBI}$ of \emph{basic logic BI-algebras}, defined by $x\wedge y=(x/y)\cdot y$.
\item The variety $\algcl{HA}$ of \emph{Heyting algebras}, defined by $xy=x\wedge y$.
\item The variety $\algcl{GA}$ of \emph{G\"odel algebras}, defined by $(x\to y)\vee(y\to x)=\top$ and $xy=x\wedge y$.
\item The variety $\algcl{MVBI}$ of \emph{many-valued BI-algebras}, defined relative to $\algcl{BLBI}$ by $(x\bi\bot)\bi\bot=x$.
\end{itemize}

Figure~\ref{4elt} shows how these and some other varieties are related to each other. However, the picture is just a subposet of the infinite lattice of subvarieties of GBI and cannot be used to deduce joins and meets of varieties.

\setlength{\arraycolsep}{8pt}
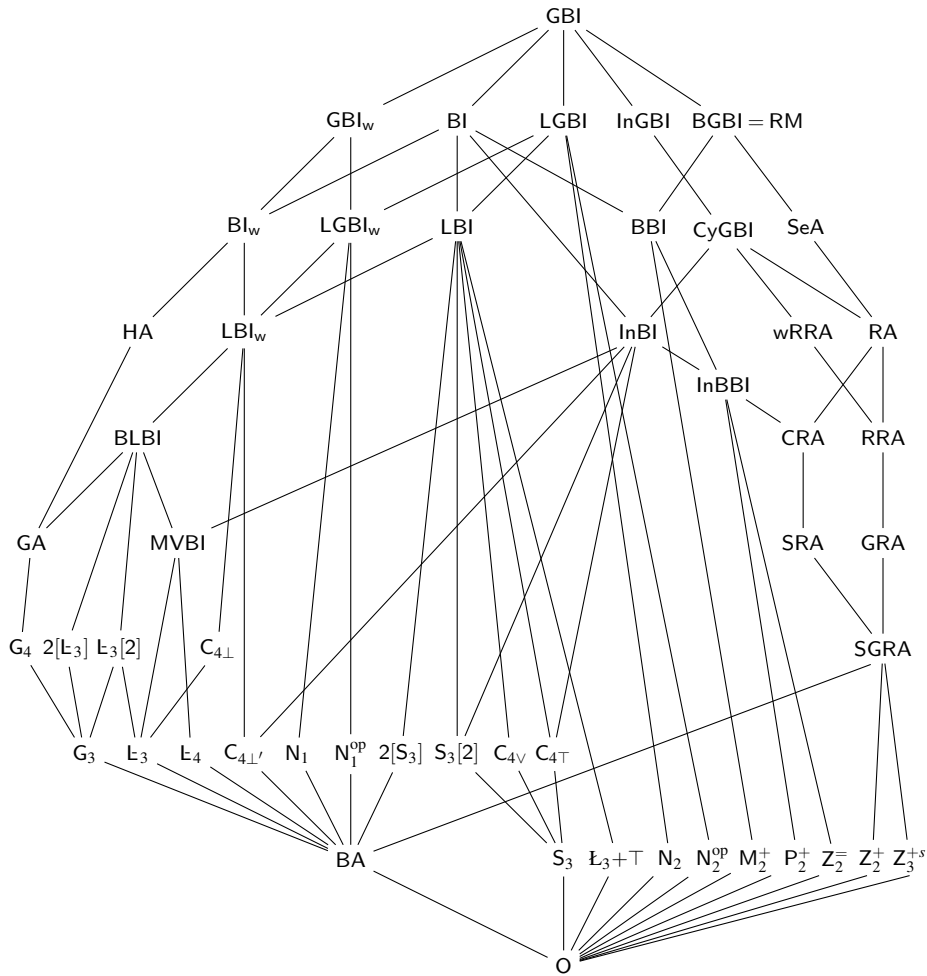
\begin{figure}
\begin{tikzpicture}[scale=0.7]
\node(T)at(2,0){$\algcl O$};
\node(BA)at(-2,2){$\algcl{BA}$};
\node(31)at(-6,4){$\Lv_3$};
\node(32)at(-7,4){$\algcl G_3$};
\node(33)at(2,2){$\algcl S_3$};
\node(1)at(-5,4){$\Lv_4$};
\node(5)at(-4,4){$\algcl C_{4\bot'}$};
\node(7)at(-3,4){$\algcl N_1$};
\node(8)at(-2,4){$\algcl N_1^\text{op}$};
\node(6)at(-4.5,6){$\algcl C_{4\bot}$};
\node(2)at(-6.35,6){$\Lv_3[\algcl 2]$};
\node(3)at(-7.35,6){$\algcl 2[\Lv_3]$};
\node(4)at(-8.2,6){$\algcl G_4$};
\node(13)at(-1.05,4){$\algcl 2[\algcl S_3]$};
\node(11)at(1,4){$\algcl C_{4\vee}$};
\node(10)at(0,4){$\algcl S_3[\algcl 2]$};
\node(9)at(1.8,4){$\algcl C_{4\top}$};
\node(12)at(3,2){$\algcl\L_3{+}\top$};
\node(14)at(4,2){$\algcl N_2$};
\node(15)at(4.8,2){$\algcl N_2^\text{op}$};
\node(M2)at(5.6,2){$\algcl M_2^+$};
\node(P2)at(6.4,2){$\algcl P_2^+$};
\node(Z21)at(7.1,1.96){$\algcl Z_2^{=}$};
\node(Z2)at(7.8,2){$\algcl Z_2^+$};
\node(Z3)at(8.5,2){$\algcl Z_3^{+s}$};
\node(MVBI)at(-5.25,8){$\mathsf{MVBI}$};
\node(GA)at(-8,8){$\mathsf{GA}$};
\node(BLBI)at(-6,10){$\mathsf{BLBI}$};
\node(GBIw)at(-2,16){$\mathsf{GBI_w}$};
\node(BIw)at(-4,14){$\mathsf{BI_w}$};
\node(LBIw)at(-4,12){$\mathsf{LBI_w}$};
\node(LGBIw)at(-2,14){$\mathsf{LGBI_w}$};
\node(LBI)at(0,14){$\mathsf{LBI}$};
\node(LGBI)at(2,16){$\mathsf{LGBI}$};
\node(GBI)at(2,18){$\mathsf{GBI}$};
\node(BBI)at(3.6,14){$\mathsf{BBI}$};
\node(BGBI)at(5,16){$\mathsf{\qquad \ BGBI=RM}$};
\node(SeA)at(6.5,14){$\mathsf{\ SeA}$};
\node(BI)at(0,16){$\mathsf{BI}$};
\node(HA)at(-6,12){$\mathsf{HA}$};
\node(RA)at(8,12){$\mathsf{RA}$};
\node(CRA)at(6.5,10){$\mathsf{CRA}$};
\node(SRA)at(6.5,8){$\mathsf{SRA}$};
\node(InGBI)at(3.5,16){$\mathsf{InGBI}$};
\node(CyGBI)at(5,13.9){$\mathsf{CyGBI}$};
\node(InBBI)at(5,11){$\mathsf{InBBI}$};
\node(InBI)at(3.4,12){$\mathsf{InBI}$};
\node(wRRA)at(6.5,12){$\mathsf{wRRA}$};
\node(RRA)at(8,10){$\mathsf{RRA}$};
\node(GRA)at(8,8){$\mathsf{GRA}$};
\node(SGRA)at(8,6){$\mathsf{SGRA}$};
\draw(T)--(BA)--(31);
\draw(BA)--(1)--(MVBI);
\draw(BA)--(5)--(LBIw);
\draw(BA)--(7)--(LGBIw);
\draw(BA)--(8)--(LGBIw);
\draw(T)--(14)--(LGBI);
\draw(T)--(15)--(LGBI);
\draw(BA)--(32);
\draw(BA)--(13);
\draw(BA)--(7.77,5.7);
\draw(T)--(33);
\draw(31)--(6)--(LBIw);
\draw(31)--(2)--(BLBI);
\draw(3)--(BLBI);
\draw(31)--(MVBI);
\draw(T)--(12)--(LBI);
\draw(32)--(2);
\draw(32)--(3);
\draw(32)--(4)--(GA);
\draw(13)--(LBI);
\draw(33)--(11)--(LBI);
\draw(33)--(10)--(LBI);
\draw(33)--(9)--(LBI);
\draw(T)--(M2)--(BBI)(InBBI)--(BBI)--(BI);
\draw(T)--(6.8,1.7)(Z21)--(InBBI);
\draw(T)--(7.7,1.7)(Z2)--(SGRA);
\draw(T)--(8.5,1.7)(Z3)--(SGRA);
\draw(T)--(5.9,1.7)(P2)--(InBBI);
\draw(SGRA)--(SRA)--(CRA)--(RA);
\draw(SGRA)--(GRA)--(RRA)--(RA);
\draw(RRA)--(wRRA)--(CyGBI)--(InGBI)--(GBI);
\draw(RA)--(CyGBI)--(InBI)--(InBBI)--(CRA);
\draw(RA)--(SeA)--(BGBI);
\draw(BI)--(InBI);
\draw(LGBIw)--(LGBI);
\draw(LBIw)--(LGBIw);
\draw(LGBIw)--(GBIw);
\draw(LBI)--(BI);
\draw(LBIw)--(BIw)--(GBIw)--(GBI);
\draw(GA)--(HA)--(BIw)--(BI);
\draw(GBI)--(BGBI)--(BBI);
\draw(LBIw)--(LBI)--(LGBI)--(GBI)--(BI);
\draw(MVBI)--(BLBI)--(LBIw);
\draw(MVBI)--(3,11.7)(InBI)--(5)(InBI)--(9)(InBI)--(10);
\draw(GA)--(BLBI);
\end{tikzpicture}

\caption{Some subvarieties of GBI ordered by inclusion. Algebras are given in Table~\ref{4eltalgs} and denote the variety they generate (in the corresponding font).}\label{4elt}
\end{figure}

Recall that an algebra is \emph{subdirectly irreducible} if its congruence lattice has a minimal nontrivial congruence, and that any algebra is a subalgebra of a product of its subdirectly irreducible homomorphic images (cf. \cite{Raftery:here}).
The subdirectly irreducible members of $\algvr V$ are denoted by $\algvr V_\text{SI}$,
hence $\algvr V=\alo{SP}(\algvr V_\text{SI})$.

Every variety $\algvr V$ is equal to $\Vo(\m A)$ for some algebra $\m A\in \algvr V$ since varieties contain countably generated free algebras. A variety is \emph{finitely generated} if it is of the form $\Vo(\m A)$ for some \emph{finite} algebra. In this case, if $\Vo(\m A)$ is congruence distributive (i.e., all members have distributive congruence lattices), then J\'onsson's Lemma implies $\Vo(\m A)_\text{SI}\subseteq \alo{HS}(\m A)$. Hence any subvariety of a finitely generated congruence distributive variety is again finitely generated and for a finite algebra $\m A$, $\Vo(\m A)$ has only finitely many subvarieties. Since we also have $\Vo(\{\m A,\m B\})=\Vo(\m A\times\m B)$, 
 finitely generated subvarieties of $\Lambda$ form a lattice ideal. In particular, the varieties of residuated lattices and GBI-algebras are congruence distributive, since they are varieties of lattice-ordered algebras and all lattices have distributive congruence lattices.

As a result we can investigate the bottom of the lattice $\Lambda$ by investigating finite subdirectly irreducible GBI-algebras. In any GBI-algebra one has $\bot\le x\backslash \bot$, hence $x\bot=\bot$. It follows that a GBI-algebra with $1=\bot$ must be trivial, hence it generates the variety $\algcl O$. 
The smallest nontrivial GBI-algebra is the 2-element Boolean algebra, with $1=\top$, $\cdot=*=\wedge$ and ${\bi} = {\to}$. Naturally this algebra generates the variety $\algcl{BA}$ of Boolean algebras.

A 3-element lattice must be linearly ordered, so we can assume $\m A=\{\bot<a<\top\}$. There are in fact 3 such algebras: The 3-element G\"odel algebra $\m G_3$ where $a*a=a$ and $1=\top$, the 3-element MV-algebra $\Lu_3$ where $a*a=\bot$ (hence $1=\top$), and the Sugihara algebra $\m S_3$ where $a*a=a=1$. The operations $\to, \bi$ are uniquely determined by the order and the monoid operation, and it is easy to check that these algebras are subdirectly irreducible.

An algebra is \emph{simple} if it has exactly two congruences, and it is \emph{strictly simple} if, in addition, it has no proper subalgebras. Using J\'onsson's Lemma it is easy to see that strictly simple algebras generate varieties that only contain $\algcl O$ as proper subvariety. Note that $\m{S}_3$ is strictly simple, $\Lu_3$ is simple but not strictly simple, and $\m{G}_3$ is subdirectly irreducible but not simple. Both $\Lu_3$ and $\m{G}_3$ have a subalgebra isomorphic to the 2-element Boolean algebra, hence they generate varieties with two proper subvarieties.

There are several methods for constructing 
 and combining residuated lattices. 
 We consider two that also apply to GBI-algebras. These constructions are used in Table~\ref{4eltalgs} to provide convenient names for some of the algebras.

\emph{Generalized ordinal sum}: This construction, denoted as $\m A[\m B]$ and described in detail by Galatos et al. \cite{GalatosR04:sl,Galatos05:au} \cite[\S~9.6.1]{GalatosJKO07}, is applicable with certain restrictions ($\m A$ must be \emph{admissible} by $\m B$).  
If $\m A$ satisfies $1=\top$ then this is the usual ordinal sum of two bounded lattice-ordered algebras.

\emph{Adding a new top}: Let $\m A$ be a GBI-algebra with $1=\top$. The algebra $\m A+\bar\top$ is defined by $\m A+\bar\top=A\cup\{\bar\top\}$ where $\bar\top$ is strictly greater than $\top$. The fusion operation is extended by $\bar\top\cdot a=a=a\cdot\bar\top$ for all $a\in A-\{1\}$. Hence $\bar\top$ is almost an identity element, except that $\bar\top1=\bar\top$. It is easy to check that the operation $\cdot$ is associative, and the residuals are definable in terms of $\cdot$ and $\wedge$. Therefore the algebra $\m A+\bar\top$ is also a GBI-algebra.

 There are exactly 20 nonisomorphic GBI-algebras with 4 elements (see Table~\ref{4eltalgs} and Figure~\ref{4elt}). The number of nonisomorphic join-preserving monoid operations on a finite distributive lattice increase rapidly: 
\begin{center}
\setlength{\tabcolsep}{10pt}
\begin{tabular}{|c|c|c|c|c|c|c|c|}\hline
$n=$&2&3&  4&    5&    6&      7&8\\\hline
GBI &1 &3&20&115&899&7782&80468\\
BI    &1 &3&16&  70&399&2261&14358\\\hline
\end{tabular}
\end{center}

\begin{table}
\setlength{\arraycolsep}{1pt}
$\begin{array}{c|cc}
\Lu_3&a&1\\\hline
a&\bot&a\\
1&a&1\\
\end{array}$
\quad
$\begin{array}{c|ccc}
\Lu_4&a&b&1\\\hline
a&\bot&\bot&a\\
b&\bot&a&b\\
1&a&b&1\\
\end{array}$
\quad
$\begin{array}{c|ccc}
\Lu_3[2]&a&b&1\\\hline
a&\bot&a&a\\
b&a&b&b\\
1&a&b&1\\
\end{array}$
\quad
$\begin{array}{c|ccc}
2[\Lu_3]&a&b&1\\\hline
a&a&a&a\\
b&a&a&b\\
1&a&b&1\\
\end{array}$
\quad
$\begin{array}{c|ccc}
\m{C}_{4\bot}&a&b&1\\\hline
a&\bot&\bot&a\\
b&\bot&\bot&b\\
1&a&b&1\\
\end{array}$
\quad
$\begin{array}{c|ccc}
\m{C}_{4\bot'}&a&b&1\\\hline
a&\bot&\bot&a\\
b&\bot&b&b\\
1&a&b&1\\
\end{array}$
\linebreak

\quad

$\begin{array}{c|cc}
\m{G}_3&a&1\\\hline
a&a&a\\
1&a&1\\
\end{array}$
\quad
$\begin{array}{c|ccc}
\m{G}_4&a&b&1\\\hline
a&a&a&a\\
b&a&b&b\\
1&a&b&1\\
\end{array}$
\quad \ 
$\begin{array}{c|ccc}
\m{N}_1&a&b&1\\\hline
a&\bot&\bot&a\\
b&a&b&b\\
1&a&b&1\\
\end{array}$
\quad
$\begin{array}{c|ccc}
\m{N}_1^\text{op}&a&b&1\\\hline
a&\bot&a&a\\
b&\bot&b&b\\
1&a&b&1\\
\end{array}$
\quad
$\begin{array}{c|ccc}
\m{C}_{4\vee}&1&a&\top\\\hline
1&1&a&\top\\
a&a&a&\top\\
\top&\top&\top&\top\\
\end{array}$
\quad
$\begin{array}{c|ccc}
\m{C}_{4\top}&1&a&\top\\\hline
1&1&a&\top\\
a&a&\top&\top\\
\top&\top&\top&\top\\
\end{array}$
\linebreak

\quad

$\begin{array}{c|cc}
\m{S}_3&1&\top\\\hline
1&1&\top\\
\top&\top&\top\\
\end{array}$
\quad
$\begin{array}{c|ccc}
\m{S}_3[2]&a&1&\top\\\hline
a&a&a&\top\\
1&a&1&\top\\
\top&\top&\top&\top\\
\end{array}$
\quad
$\begin{array}{c|ccc}
2[\m{S}_3]&a&1&\top\\\hline
a&a&a&a\\
1&a&1&\top\\
\top&a&\top&\top\\
\end{array}$
\quad
$\begin{array}{c|ccc}
\m{N}_2&a&1&\top\\\hline
a&a&a&a\\
1&a&1&\top\\
\top&\top&\top&\top\\
\end{array}$
\quad
$\begin{array}{c|ccc}
\m{N}_2^\text{op}&a&1&\top\\\hline
a&a&a&\top\\
1&a&1&\top\\
\top&a&\top&\top\\
\end{array}$
\quad
$\begin{array}{c|ccc}
\Lu_3{+}\top&a&1&\top\\\hline
a&\bot&a&a\\
1&a&1&\top\\
\top&a&\top&\top\\
\end{array}$
\linebreak

\quad

\indent \qquad \qquad \qquad$\begin{array}{c|ccc}
\m P_2^+&1&0&\top\\\hline
1&1&0&\top\\
0&0&\bot&\top\\
\top&\top&\top&\top\\
\end{array}$
\quad \  \ 
$\begin{array}{c|ccc}
\mathbb Z_2^+&1&0&\top\\\hline
1&1&0&\top\\
0&0&1&\top\\
\top&\top&\top&\top\\
\end{array}$
\quad \  \ 
$\begin{array}{c|ccc}
\mathbb Z_3^{+s}&1&0&\top\\\hline
1&1&0&\top\\
0&0&\top&\top\\
\top&\top&\top&\top\\
\end{array}$
\quad \  \ 
$\begin{array}{c|ccc}
\m M_2^+&1&0&\top\\\hline
1&1&0&\top\\
0&0&0&\top\\
\top&\top&\top&\top\\
\end{array}$
\quad \ 
$\begin{array}{c|ccc}
2\times 2&1&0&\top\\\hline
1&1&\bot&1\\
0&\bot&0&0\\
\top&1&0&\top\\
\end{array}$
\caption{Multiplicative operation tables for all 3- and 4-element GBI-algebras ($\bot$ not listed since $\bot x=\bot=x\bot$). The order is linear, except for the last row where it is Boolean.\label{4eltalgs}}
\end{table}

\noindent
Any finite distributive residuated lattice is the reduct of a GBI-algebra, but this observation does not extend to the infinite setting. 
  For example, pick a bounded distributive lattice 
 that is not the reduct of a Heyting algebra, and that has an atom $1$. On such a lattice one can define a fusion operation by 
$$x\cdot y=\begin{cases}1&\text{ if }x\ne\bot\ne y\\
\bot &\text{ otherwise}
\end{cases}$$ and check that it is a monoid operation that is residuated.




\section{Semantics via Duality} \label{RelSem}

Heyting algebras and (G)BI-algebras provide algebraic semantics for intuitionistic logic and (noncommutative) bunched implication logic, respectively (cf. \S~\ref{sec:logic} and especially Remark \ref{rem:alg}). However the algebras that are of interest can be rather large, or they may have quite complicated order structure. Since they have distributive lattice reducts, it is useful to consider smaller or more concrete combinatorial structures from which the lattice order can be recovered. Considering the categories \algcl{HA} and \algcl{GBI}, with homomorphisms as morphisms, one would like to have equivalent or dually equivalent categories. For \algcl{HA} there is a well-developed duality theory based on Esakia spaces, and we briefly recall the relevant details here. Adding a suitable ternary Kripke relation extends this duality to GBI-algebras as well as to involutive GBI-algebras. Finally we consider a relational semantics based on residuated frames, since this is closely related to the proof theory that we present in \S~\ref{ProofTheory}.

Before presenting the topological dualities, we first consider Birkhoff's duality for finite distributive lattices, and its extension to complete and perfect distributive lattices. Note that for an element $a$ in a lattice, $\bigvee\{x\mid x<a\}$ always exists, and is either $a$ or a dual cover of $a$. In the latter case, $a$ is said to be \emph{completely join-irreducible},
and its dual cover is denoted by $a_*$.
The set of completely join-irreducible elements of a lattice $\m L$ is denoted by $J(\m L)$. Dually, a \emph{completely meet-irreducible} element $b$ satisfies $b\prec\bigwedge\{x\mid b<x\}=b^*$, and the set of all such elements is denoted $M(\m L)$.

A lattice is \emph{complete} if all joins and meets exist. Even for a complete lattice, $J(\m L)$ and/or $M(\m L)$ may be empty, as happens for the unit interval of real numbers.

A lattice is \emph{join-perfect} if every element is the join of completely join-irreducible elements, it is \emph{meet-perfect} if very element is the meet of completely meet-irreducible elements, and it is \emph{perfect} if both conditions hold. For example, a Boolean algebra is join-perfect if it is atomic (= every element is a join of atoms, defined as minimal non-zero elements), and every finite lattice is perfect. For a Boolean algebra, being 
join-perfect is equivalent to being meet-perfect since complementation is a dual isomorphism. However, even for complete distributive lattices this is not the case, as can be seen from the join-perfect distributive lattice $\mathbb N\times\mathbb N$ completed with a top element, since it has no completely meet-irreducible elements.
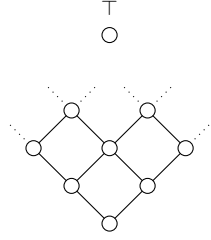
\begin{figure}[t!]
\begin{center}
\begin{tikzpicture}[scale=.5, every node/.style={circle, draw, fill=white, inner sep=2pt}]
\draw(0,5)node[label={$\top$}]{};
\draw[dotted](-2,2)--(-2.7,2.7);
\draw(0,0)--(-1,1)--(-2,2);
\draw[dotted](-1,3)--(-1.7,3.7);
\draw(1,1)--(0,2)--(-1,3);
\draw[dotted](1,3)--(.3,3.7);
\draw(2,2)--(1,3);
\draw[dotted](2,2)--(2.7,2.7);
\draw(0,0)node{}--(1,1)node{}--(2,2)node{};
\draw[dotted](1,3)--(1.7,3.7);
\draw(-1,1)node{}--(0,2)node{}--(1,3)node{};
\draw[dotted](-1,3)--(-.3,3.7);
\draw(-2,2)node{}--(-1,3)node{};

\end{tikzpicture}
\end{center}

\caption{A complete join-perfect distributive lattice that is not meet-perfect.
The dual lattice is a Heyting algebra that is not join-perfect.\label{perfect}}
\end{figure}

\begin{lemma}
Any join-perfect Heyting algebra is also meet-perfect, hence a perfect lattice.
\end{lemma}
\begin{proof} Suppose $a$ is an element in a join-perfect Heyting algebra $\m A$, and let $S=\{x\in M(\m A)\mid a\le x\}$. If $a\ne\bigwedge S$ then there exists a lower bound $b$ of $S$ such that $b\nleq a$. Since $b$ is a join of completely join-irreducibles, there exists $c\in J(\m A)$ such that $c\le b$ and $c\nleq a$. Let $d=c\to c_*$, where $c_*$ is the unique dual cover of $c$. We show that $d\in S$ then it follows that $d\ge b$, which is a contradiction.

From $c\nleq a$ we get $c\wedge a\le c_*$, hence $a\le c\to c_*=d$. To see that $d\in M(\m A)$, note that if $y>d$ then $y\nleq c\to c_*$, therefore $c\wedge y\nleq c_*$ and hence $c\le y$. It follows that $c\le\bigwedge\{y\mid d<y\}$, and since $c\nleq d$ we conclude that $d\prec d^*=\bigwedge\{y\mid d<y\}$.
\hfill\qed\end{proof}
The converse of this lemma does not hold, since for example the dual of $\mathbb N\times\mathbb N$ with a bottom added is a (complete) meet-perfect Heyting algebra that is not join-perfect.

For complete Heyting algebras, the notion of being perfect is equivalent to being a doubly algebraic lattice (i.e., a complete lattice in which every element is the join of compact elements and the meet of co-compact elements).

\subsection{Semantics and Duality for Heyting Algebras}
Tarski proved that complete and atomic Boolean algebras are isomorphic to powerset algebras, and that complete homomorphisms between Boolean algebras are induced by functions (in the opposite direction) between their sets of atoms. In a nutshell this is the categorical duality between $\algcl{caBA}$, the category of complete and atomic Boolean algebras, 
and $\algcl{Set}$, the category of sets.

Birkhoff observed that every finite distributive lattice $\m L$ is isomorphic to the set of downward closed subsets of $(J(\m L),\le)$, with intersection and union as lattice operations. Alternatively one can take as the starting point  the set $F_\text p(\m L)$ of prime filters: 
  in the finite case, $(F_\text p(\m L),\subseteq)\cong (J(\m L),\ge)$, hence $\m L$ is isomorphic to the upward closed subsets of $F_\text p(\m L)$. 

Actually, it is well known (and easy to see) that for any poset $\m P = (P,\le)$, the collection $\ups{\m P}$ of upward closed sets is a complete distributive lattice under $\bigcap$ and $\bigcup$. 
 Since 
  $\bigcup$ distributes over $\bigcap$,  $\ups{\m P}$ is in fact a complete Heyting algebra, with $U\to V=P-{\downarrow}(U-V)$. The completely join-irreducible elements of $\ups{\m P}$ are the principal upsets, so $(\ups{\m P},\bigcap,\bigcup)$ is join-perfect (hence perfect) and $\m P\cong J(\ups{\m P})$ via the map $p\mapsto {\uparrow}p$. Likewise, for a complete perfect distributive lattice $\m L$, $L\cong \ups{J(\m L)}$. This is a brief outline of the categorical duality between complete and perfect distributive lattices with complete homomorphisms, and posets with order-preserving maps as morphisms. However, as noted earlier $\ups{\m P}$ is a complete Heyting algebra, and to get a categorical duality, one has to modify the notion of morphism in the semantic category of posets. Rather than allowing all order-preserving maps as morphisms, we restrict to so-called \emph{p-morphisms}, which are maps such that ${\uparrow}f(p) = f[{\uparrow}p]$
for all $p\in P$. (Here $f[X]=\{f(x):x\in X\}$ denotes the forward image of $X$ under $f$. Logicians may note that ${\uparrow}f(p)\supseteq f[{\uparrow}p]$ is the ``forth" condition, i.e. order preservation, and ${\uparrow}f(p)\subseteq f[{\uparrow}p]$ is the ``back" condition of modal p-morphisms.)

With this assumption it follows that $f^{-1}[U\to V]=f^{-1}[U]\to f^{-1}[V]$ for all $U,V\in \ups{\m P}$. Hence a p-morphism $f:\m P\to \m Q$ gives rise to a complete Heyting algebra homomorphism $f^{-1}:\ups{\m Q}\to \ups{\m P}$. Conversely, a complete Heyting algebra homomorphism $h:\m A\to \m B$ corresponds to a p-morphism $J(h):J(\m B)\to J(\m A)$ given by $J(h)(b)=\bigwedge h^{-1}[{\uparrow}b]$.

In summary, the functors $J:\algcl{cpHA}\leftrightarrow\algcl{pPos}:\upstx$ give the duality between the category of complete perfect Heyting algebras and the category of posets with p-morphisms.

Recall that a \emph{Priestley space} $(P,\le,\tau)$ is a poset  $(P,\le)$ with a compact topology $\tau$ on $P$ that has a base of clopen sets and for all $p\nleq q$ in $P$ there exists a clopen upset $U$ such that $p\in U$ and $q\notin U$. By the well-known Priestley duality \cite{DP:book} \tlnt{what is [DP]? Davey and Priestley?} the category \algcl{BDL} of bounded distributive lattices with homomorphisms is
dually equivalent to the category \algcl{Pri} of Priestley spaces with order-preserving continuous maps.

The functor $F_\text p:\algcl{BDL}\to\algcl{Pri}$ maps a bounded distributive lattice $\m L$ to the Priestley space $(F_\text p(\m L),\subseteq,\tau)$. As before, $F_\text p(\m L)$ is the set of all prime filters of $\m L$, ordered by inclusion, and the topology $\tau$ is generated by the subbasis $\{\mathcal F_a:a\in L\}\cup\{F_\text p(\m L)-\mathcal F_a:a\in L\}$ where $\mathcal F_a=\{F\in F_\text p(\m L):a\in F\}$. For a homomorphisms $h:L\to M$, the map $F_\text p(h):F_\text p(M)\to F_\text p(\m L)$ given by $F_\text p(h)(G)=h^{-1}[G]$ is a continuous and order-preserving.

The functor $U_\text c:\algcl{Pri}\to\algcl{BDL}$ maps a Priestley space $\m P = (P,\le,\tau)$ to the set $U_\text c(\m P)$ of clopen upsets, which is a bounded distributive lattice under intersection and union. For a continuous order-preserving map $f:(P,\le,\tau)\to (Q,\le,\tau')$, the map $U_\text c(f):U_\text c(\m Q)\to U_\text c(\m P)$ given by $U_\text c(f)(V)=f^{-1}[V]$ is a bounded distributive lattice homomorphism.

An \emph{Esakia space} is a Priestley space $(P,\le,\tau)$ such that if $U\subseteq P$ is clopen then ${\downarrow}U$ is also clopen. The restriction of Priestley duality to Heyting algebras is given by the following result.

\begin{theorem}\ 
\begin{itemize}
\item A bounded distributive lattice $\m L$ is a Heyting algebra if and only if the corresponding Priestley space $F_\text p(\m L)$ is an Esakia space. 
\item For Heyting algebras $\m A,\m B$ a bounded distributive lattice homomorphism $h:A \to B$ preserves the Heyting implication if and only if $F_\text p(h)$ is a p-morphism.
\end{itemize}
\end{theorem}

\subsection{Semantics and Duality for GBI-algebras} \label{DualityGBI}


To extend this duality to GBI-algebras, 
  we need a ternary relation $\circ: P^2\to\mathcal P(P)$ and a unary relation $E\subseteq P$ satisfying certain conditions to an Esakia space $(P,\le,\tau)$. The definitions and results of this subsection are based on  Urquhart \cite{Urquhart1996} and 
   Galatos \cite{Galatos00}; cf. also more recent references motivated by CS applications of BI  \cite{DochertyP17:ilgl,DochertyP17:stone,BrotherstonV14}, where similar  dualities or representation theorems have been developed independently.  \tlnt{What is [Galatos2000]?} We extend $\circ$ to a binary
operation on $\mathcal P(P)$ by
$$
X\circ Y=\bigcup_{x\in X,y\in Y}x\circ y,\quad x\circ Y=\{x\}\circ Y,\quad X\circ y=X\circ \{y\}
$$
and let 
$$X\cdot Y={\uparrow}(X\circ Y), \ X\backslash Y=\{z:X\circ z\subseteq Y\}, \ X/Y=\{z:z\circ Y\subseteq X\}.$$
Now define 
$\m P=(P,\le,\circ,E,\tau)$ to be a \emph{GBI-space} if $(P,\le,\tau)$ is an Esakia space
and the following properties hold:
\begin{esakia}
\item $(x\circ y)\circ z=x\circ(y\circ z)$,
\item $E$ is a clopen upset, and $E\circ U=U=U\circ E$ for any clopen upset $U$,
\item $x'\le x, y'\le y, z\le z'$, and $z\in x\circ y\implies z'\in x'\circ y'$,
\item for all clopen upsets $U,V$ of $P$, $U\circ V$, $U\backslash V$ and $U/V$ are clopen and
\item $z\notin x\circ y$ implies there exist clopen upsets $U,V$ such that $x\in U,y\in V$ and $z\notin U\circ V$.
\end{esakia}
For GBI-spaces $\m P,\m Q$ a \emph{GBI-p-morphism} $f:\m P\to\m Q$ is a continuous Heyting algebra p-morphism $f:(P,\le,\tau)\to (Q,\le,\tau')$ that satisfies\footnote{The need for items 8. and 9. was pointed out by Docherty and Pym, see \cite{DochertyP17:ilgl,DochertyP17:stone}. 
Here we use item 8. as first introduced by Urquhart \cite{Urquhart1996} for relevant implication $x\to y=y/x$, and item 9. is the corresponding version for $\backslash$. }
\begin{esakia}
\item $z\in x\circ_\m Py \implies f(z)\in f(x)\circ_\m Qf(y)$,
\item $f(z)\in u\circ_\m Q v\implies \exists\,x,y\in P\ (u\le f(x)$, $v\le f(y)$ and $z\in x\circ_\m Py)$,
\item $w\in f(x)\circ_\m Q v \implies \exists\, y,z\in P\ (v\le f(y)$, $f(z)\le w$ and $z\in x\circ_\m Py)$,
\item $w\in u\circ_\m Q f(y)\implies \exists\, x,z\in P\ (u\le f(x)$, $f(z)\le w$ and $z\in x\circ_\m Py)$, and
\item $E_\m P=f^{-1}[E_\m Q]$.
\end{esakia}
The category of GBI-spaces with GBI-p-morphisms is denoted by \algcl{GBIp}.
Finally we define the functor $F_\text p:\algcl{GBI}\to\algcl{GBIp}$
by $$F_\text p(\m A)=(F_\text p(A),\subseteq,\circ,\mathcal F_1,\tau)$$ where
\begin{itemize}
\item $\circ:F_\text p(A)^2\to \mathcal P(F_\text p(A))$ is given by $F\circ G=\bigcap_{a\in F, b\in G}\mathcal F_{ab}$ and 
\item $(F_\text p(A),\subseteq,\tau)$ is the Esakia
space of the Heyting algebra reduct of $\m A$. 
\end{itemize}

Conversely, the functor $U_\text c:\algcl{GBIp}\to\algcl{GBI}$ is defined by $$U_\text c(\m P)=(U_\text c(P),\cap,\cup,\to,\emptyset,P,\cdot,\backslash,/,E).$$ On morphisms, these two functors act the same way as for Heyting algebras and Esakia spaces.

\begin{theorem}\ 
\begin{upbroman}
\item For any GBI-algebra $\m A$, $F_\textup p(\m A)$ is a GBI-space.
\item For any GBI-space $\m P$, $U_\textup c(\m P)$ is a GBI-algebra.
\item $U_\textup c(F_\textup p(\m A))\cong \m A$ and $F_\textup p(U_\textup c(\m P))\cong \m P$.
\item On morphisms, the functors $U_\textup cF_\textup p$ and $F_\textup pU_\textup c$ are naturally isomorphic to the respective identity functors, hence the category of GBI-algebras is dually equivalent to the category of GBI-spaces.
\end{upbroman}
\end{theorem}
\begin{proof}
Since the details of Priestley and Esakia duality are well known \cite{DP:book, DG}, \tlnt{is [DP] Davey and Priestley?} we verify only the properties related to $\cdot$, $\backslash$, $/$, $1$, $\circ$ and $E$. We first state some auxiliary facts that are easy to check. For filters $F,G$ of $\m A$ we let $FG=\{ab:a\in F,b\in G\}$.

\noindent Claim 1. ${\uparrow}(FG)$ is a filter of $\m A$. Pf: For any $a,b\in{\uparrow}(FG)$ there exist $a',b'\in F,a'',b''\in G$ such that $a'a''\le a$ and $b'b''\le b$. Hence $a'\wedge b'\in F, a''\wedge b''\in G$ and $a\wedge b\ge  a'a''\wedge b'b''\ge(a'\wedge b')(a''\wedge b'')\in FG$.

\noindent Claim 2. For $F,G,H\in F_\text p(A)$, $H\in F\circ G$ if and only if $FG\subseteq H$. Pf: Follows from the definition $F\circ G=\bigcap_{a\in F,b\in G} \mathcal F_{ab}$.

\noindent Claim 3. For filters $F,G$ of $\m A$ and $FG\subseteq H\in F_\text p(A)$, there exist $F',G'\in F_\text p(A)$ such that $F'G,FG',F'G'\subseteq H$. 

\begin{upbroman}
\item We prove properties 1.-5. of GBI-spaces. For 1. let $F,G,H,K\in F_\text p(A)$ and assume $K\in (F\circ G)\circ H$. Then there exists $M\in F\circ G$ such that $K\in M\circ H$. By Fact 2, $FG\subseteq M$ and $MH\subseteq K$. Consider the filter $N={\uparrow}(GH)$ and note that $FN\subseteq K$ since $\cdot$ is associative. By Fact 3, there exists $N'\in F_\text p(A)$ such that $N\subseteq N'$ and $FN'\subseteq K$. Since we also have $GH\subseteq N'$ it follows that $K\in F\circ N'$ and $N'\in G\circ H$, whence $K\in F\circ (G\circ H)$. The converse is similar.

Recall that $\mathcal F_1=\{F\in F_\text p(A):1\in F\}$, which is a subbasic clopen of $\tau$ and is also an upset in the inclusion order on prime filters. Let $\mathcal U$ be a clopen upset of prime filters. By Priestley duality $\mathcal U=\mathcal F_a$ for some $a\in A$, hence $\mathcal F_1\circ \mathcal U=\mathcal F_1\circ \mathcal F_a=\mathcal F_{1a}=\mathcal F_a=\mathcal U$. The other equality of 2. is similar.

For 3. let $F'\subseteq F$, $G'\subseteq G$, $H\subseteq H'$ and $H\in F\circ G$. Then $FG\subseteq H\subseteq H'$ so $F'G'\subseteq FG\subseteq H'$ which implies $H'\in F'\circ G'$.

Let $\mathcal U,\mathcal V$ be clopen upsets of prime filters. By Priestley duality there exist $a,b\in A$ such that $\mathcal U=\mathcal F_a$ and $\mathcal V=\mathcal F_b$. 
We prove that $\mathcal F_a\circ \mathcal F_b=\mathcal F_{ab}$, $\mathcal F_a\backslash\mathcal F_b=\mathcal F_{a\backslash b}$ and $\mathcal F_a/\mathcal F_b=\mathcal F_{a/b}$, then 4. holds. Let $H\in \mathcal F_a\circ \mathcal F_b$, so $H\in F\circ G$ where $a\in F$ and $b\in G$ for prime filters $F,G$. Then $ab\in FG\subseteq H$ by Claim 2, whence $H\in \mathcal F_{ab}$. Conversely, if $H\in \mathcal F_{ab}$ then ${\uparrow}a{\uparrow}b\subseteq H$. By Claim 3 there exist prime filters $F,G$ such that ${\uparrow}a\subseteq F$, ${\uparrow}b\subseteq G$ and $FG\subseteq H$, so $H\in F\circ G\subseteq \mathcal F_a\circ \mathcal F_b$.

Next, note that $\mathcal F_a\circ \mathcal F_{a\backslash b}=\mathcal F_{a(a\backslash b)}\subseteq \mathcal F_b$ since $a(a\backslash b)\in F$ implies $b\in F$. Therefore $\mathcal F_{a\backslash b}\subseteq \mathcal F_a\backslash\mathcal F_b$. For the opposite inclusion, let $G\in \mathcal F_a\backslash\mathcal F_b$ whence $\mathcal F_a\circ G\subseteq \mathcal F_b$. We aim to show that $G\in\mathcal F_{a\backslash b}$, i.e., $a\backslash b\in G$. 
If $a=\bot$ this holds trivially, so assume $a\ne\bot$ and let $F\in\mathcal F_a$.
Then for any $H\in F\circ G$ we have $H\in\mathcal F_b$, or equivalently,
$FG\subseteq H$ implies $b\in H$. By Claim 1, ${\uparrow}(FG)$ is a filter, 
so $b\in{\uparrow}(FG)$ (otherwise we could extend this to a prime filter $H$ that does not contain $b$). We conclude that $ax\le b$ for some $x\in G$, whence $x\le a\backslash b\in G$. The argument for $a/b$ is similar.

Finally, for 5. let $F,G,H$ be prime filters such that $H\notin F\circ G$. Then ${\uparrow}(FG)\not\subseteq H$, so there exist $a\in F,b\in G$ such that $ab\le c\notin H$. Then $F\in\mathcal F_a,G\in\mathcal F_b$ and $\mathcal F_a\circ\mathcal F_b=\mathcal F_{ab}$, hence $H\notin\mathcal F_a\circ\mathcal F_b$.

\item For a GBI-space $\m P$, we take for granted that $(U_\text c(P),\cap,\cup,\to,\top,\bot)$ is a Heyting algebra. By Property 1 the ternary relation $\circ$ lifts to an associative operation on $\mathcal P(P)$. By Property 2 the set $E$ is a unit element for this lifted operation when restricted to clopen upsets. For $U,V\in U_\text c(P)$ Property 3 implies that $U\circ V$ is also an upset, and by Property 4 it will be clopen. Hence $(U_\text c(P),\circ,E)$ is a monoid. 

By definition $x\in U/V\iff x\circ V\subseteq U$, and for $x\le w$ we have $w\circ V\subseteq x\circ V$ since $z\in w\circ V$ implies $z\in w\circ y$ for some $y\in V$, so $z\in x\circ y$ by Property 3. Therefore $w\circ V\subseteq U$, or equivalently $w\in U/V$, which shows that $U/V$ is an upset. By Property 4 it is also clopen, hence $U/V\in U_\text c(P)$. The argument for $U\backslash V$ is similar. 

Now suppose $W\subseteq U/V$. This is equivalent to $w\circ V\subseteq U$ for all $w\in W$ which in turn is equivalent to $W\circ V\subseteq U$. Hence $U_\text c(\m P)$ is a GBI-algebra.

\item By Priestley duality, the map $a\mapsto \mathcal F_a$ is a bounded distributive lattice isomorphism from $\m A$ to $U_\text c(F_\text p(\m A))$, and since $\to$ is uniquely determined by the lattice, it is also a Heyting algebra isomorphism. Furthermore, by the proof of Property 4 in (i) this map preserves $\circ$, $\backslash$, and $/$. For a GBI-space $\m P$ consider the map $f:P\to F_\text p(U_\text c(\m P)$ defined by $f(x)=\{U\in U_\text c(P):x\in U\}$. This is an isomorphism of the Esakia spaces and it suffices to check that $z\in x\circ y$ if and only if $f(z)\in f(x)\circ f(y)$.

\item It remains to check that the functor $F_\text p$ sends a GBI homomorphisms 
to a GBI-p-morphism, and the functor $U_\text c$ does the reverse. For
details of this argument we refer the reader to Urquhart \cite{Urquhart1996} (Thm 3.5).
The categorical duality then follows by restricting the duality for Heyting algebras and Esakia spaces to GBI-homomorphisms and GBI-p-morphisms.
\hfill\qed
\end{upbroman}

\end{proof}

\section{Decidability Issues} \label{Decidability}

\subsection{Positive Decidability Results} \label{sec:posdec}

Let us begin with a result that can be derived from the proof-theoretic framework described in Section \ref{ProofTheory} and developed by the first author in collaboration with Nick Galatos; see the discussion therein (also for related references like \cite{Kozak09}). Let \algcl{nGBI} be the variety of nonassociative GBI-algebras, i.e., defined by the axioms of \GBI\ but without
the axiom of associativity of fusion; see, e.g., Galatos and Ono \cite{GalatosO10:apal} or Galatos et al. \cite[Ch. 2.3.1]{GalatosJKO07} for more information on non-associative substructural logics (cf. also Rem. \ref{rem:nonassoc} and \cite{DochertyP17:ilgl}).

\begin{theorem}
The equational theory of \algcl{nGBI} and \BI\ is decidable.
\end{theorem}

Decidability of \BI\ is proved by Galmiche et al. \cite{GalmicheMP05:mscs}. The proof in that paper uses specific, tailor-made techniques. We are not aware of any reference prior to Galatos and Jipsen \cite{GalatosJ} claiming decidability for \algcl{nGBI} or \GBI.\footnote{\label{ft:ramaerror} The
proof of decidability of \GBI\ in Galatos and Jipsen \cite{GalatosJ} appears to have
an issue with the complexity measure, as pointed out by R. Ramanayake. Decidability of \algcl{nGBI} has also been proved by Docherty and Pym \cite{DochertyP17:ilgl}, independently of the same result in Galatos and Jipsen \cite{GalatosJ}.} 

It seems that at the moment, there is no systematic investigation of complexity of these equational theories; some upper bounds are claimed by Ramanayake \cite{Ramanayake16}. Another problem which seems presently open is the question of decidability of non-boolean involutive varieties like $\algcl{InGBI}$ and $\algcl{CyGBI}$.



Obviously, there are subareas of $\Lambda$ 
 which allow nice decidability results for whole classes of varieties, in fact going beyond purely equational theory. Apart from finitely generated varieties like \BA\ and other ones in the bottom area of  Figure \ref{4elt}, we have numerous decidability results for subvarieties of \HA. 
  In fact, many of them enjoy rather low complexity, from PSPACE for \HA\ itself to NP for \GA\ and its subvarieties. These results are well described in standard monographs  
    \cite{ChagrovZ97:ml}.

However,  in other parts of $\Lambda$ positive decidability results are much less common. Moreover, things look even worse when one goes beyond purely equational theory. Let us discuss the two issues separately. 

\subsection{Subvarieties with Undecidable Equational Theory} \label{UndecEq}

As it turns out, undecidability results abound for \algcl{BGBI}\ (also known as \algcl{RM}) and its extensions, in particular \BBI. Powerful  general undecidability results for such varieties were established by algebraic  logicians in the 1990's and summarized in references like   Andr\'eka et al. \cite{AndrekaKNSS96:csli} or earlier Kurucz et al. \cite{KuruczNSS95:jolli}.  Here is a result most relevant for our purposes:

\begin{theorem}
\label{th:HungarianUndec}
A  variety $\algvr{V} \subseteq \algcl{BGBI}$   
 is undecidable whenever 
\begin{itemize}
\item there exists an infinite algebra $\m A \in \algvr V$ containing a $(\cdot,\backslash,/)$-subreduct whose universe is an antichain in $\m A$ or
\item every finite group is semigroup-embeddable into some $\m A \in \algvr V$ or
\item  for any $n \in \nat$ there exists a product $\m G$ of $n$ nontrivial finite groups and a semigroup-embedding $f: G \to A$ into (the semigroup-reduct of) some $\m A \in \algvr{V}$ s.t. $f[G]$ is an antichain in $\m A$.
\end{itemize} 
\end{theorem}

\begin{proof}
Follows from Kurucz et al. \cite[Th. 8]{KuruczNSS95:jolli}, \cite[Th. 3.6]{AndrekaKNSS96:csli}.
\hfill\qed\end{proof}

\begin{corollary} \label{cor:HungarianUndec}
Equational theories of \algcl{BGBI}=\algcl{RM}, \algcl{BBI}=\algcl{CRM}, \algcl{IRM}, \algcl{RA}, \algcl{ICRM}, \algcl{IRA}, \algcl{CRA}, \algcl{RRA}, \algcl{GRA}, \algcl{SRA}\ and \algcl{SGRA} are undecidable. 
\end{corollary}
\noindent
For most of these classes, this corollary is explicitly stated as Corollary 5.4 in  Andr\'eka et al. \cite{AndrekaKNSS96:csli} or Corollary 8.1 in Kurucz et al. \cite{KuruczNSS95:jolli}. Of course, for varieties like \algcl{RA} and \algcl{RRA} undecidability of the equational theory was established much earlier by Tarski, in fact claimed already in 1941 \cite{Tarski41:jsl}. 

These results were somehow overlooked by the BI community. Much more recently, overlapping undecidability results for subvarieties of \BBI\ have been obtained by Brotherston and Kanovich \cite{BrotherstonK10:lics,BrotherstonK14:jacm} and by Larchey{-}Wendling and Galmiche \cite{Larchey-WendlingG10:lics,Larchey-WendlingG13:tocl}. 
   However,  we repeatedly stated in \S~\ref{sec:intro} and \S\S~\ref{ModelsBI}--\ref{ModelsGBI}  that 
     it is natural to focus on concrete models, particularly \emph{memory} and \emph{heap models} (\S~\ref{HeapModel}). 
Brotherston and Kanovich \cite{BrotherstonK14:jacm}  prove that 
 subvarieties of \BBI\ generated by such 
  models are undecidable; similar undecidability results apply to even the simplest allocation/effect algebras (cf. \S~\ref{Effect}). An alternative, somewhat more semantic proof  is provided by Demri and Deters \cite[Theorem 4.4]{DemriD15:jancl}.  It is not immediately clear how to adjust 
  algebraic proofs quoted above to cover such varieties: 
 Andr\'eka et al.  \cite{AndrekaKNSS96:csli} follow  Urquhart  \cite{urquhart84:jsl,urquhart95:au}, Freese \cite{Freese80:ams} and Lipshitz \cite{Lipshitz74:ams} in using the technique of von Neumann's \cite{vonNeumann60:book} $n$-frames, which originated in projective geometries and is applicable to a wide class of varieties, but rather heavy on the technical side. In contrast, Brotherston and Kanovich \cite{BrotherstonK14:jacm} or Demri and Deters \cite{DemriD15:jancl} use a natural strategy of encoding Minsky machines, tailored for the intended models, and more readily understandable to CS researchers.


\begin{rem}\label{rem:undecvar}
It is important to mention here that the separation logic community not only tends to be interested in concrete models, but  also typically extends the syntax with entities allowing reasoning about e.g. heap structure and program values. Moreover, 
BI and its extensions are not considered in isolation, but are of interest mostly as the assertional core of proper separation logic, i.e., 
 a suitable language of \emph{Hoare triples} (\S~\ref{SeparationLogic}). 
  On the other hand, the assertion language hardly ever allows uninterpreted algebraic variables, which also limits direct applicability of (un)decidability results discussed here. 
 Demri and Deters \cite{DemriD15:jancl} provide an overview of positive and negative results for the assertion language  of separation logic. We will return to the issue  in \S~\ref{sec:DecRev}.  
\end{rem} 

\begin{rem} \label{rem:nonassoc}
Finally, let us note that these undecidability results heavily rely on associativity. The  concluding section of Andr\'eka et al.  \cite{AndrekaKNSS96:csli}  or, even more relevantly, the chapter by Mikul\'{a}s \cite{Mikulas96:csli} in the same volume \cite{marx:arro96} show that 
 positive decidability results are available even for boolean \algcl{nGBI} and its relatives. An explanation  of practical interest in such formalisms 
 can be found in the recent work of Collinson et al. \cite{CollinsonMDP14:jlc}; in fact, this reference rediscovers a variant of system called CARL by Mikul\'{a}s \cite{Mikulas96:csli}. The intuitionistic variant (in fact equivalent to the equational theory of \algcl{nGBI} \cite{GalatosJ}; cf. also \cite{GalatosO10:apal}, \cite[Ch. 2.3.1]{GalatosJKO07})  is motivated in a similar setting by Docherty and Pym \cite{DochertyP17:stone}.
 \end{rem}

\subsection{Undecidability of Quasi-Equational Theories} 

Obviously, all the subvarieties with undecidable equational theory discussed in \S~\ref{UndecEq} have \emph{a fortiori} undecidable quasi-equational theories. Nevetheless, even having decidable equational theory does not guarantee positive results here (see, however, Remark \ref{rem:posfep}).

\begin{theorem}
The quasi-equational theory of \BI\  is undecidable.
\end{theorem} 

\begin{proof}[Sketch]
Galatos \cite{Galatos02:oas}, following a strategy similar to that of Urquhart  \cite{urquhart84:jsl,urquhart95:au}, Freese \cite{Freese80:ams} and, earlier still, of Lipshitz \cite{Lipshitz74:ams} shows that the word problem for \algcl{DRL} (\emph{distributive residuated lattices}, cf. \cite[\S~3.5.3]{GalatosJKO07}) is undecidable. As will be shown in \S~\ref{ProofTheory}, 
 this class consist precisely of subreducts of \BI\ which implies the result. 
\hfill\qed\end{proof}

\begin{rem}
The same technique can be used to show directly the undecidability of the quasiequational theory of  \BBI: just replace \algcl{DRL} with  \algcl{CDRL} (commutative distributive residuated lattices; cf. \S~\ref{ProofTheory}). But in the boolean setting, Corollary \ref{cor:HungarianUndec} provides a stronger result anyway. The general idea of using von Neumann's $n$-frames is central to both proofs.
\end{rem}

\noindent
Recall that the \emph{finite embeddability property} (cf., e.g., \cite{Blok2002}) for finitely axiomatizable (quasi)varieties implies decidability of universal theory---and, a fortiori, quasi-equational theory. Thus we obtain

\begin{corollary}
\GBI\ and \BI\ do not have the finite embeddability property.
\end{corollary}
For \BBI, we have an even stronger result that follows from Corollary \ref{cor:HungarianUndec}:

\begin{corollary}
\BBI\ does not have the finite model property.
\end{corollary}




\noindent
Still, as pointed out in Remark \ref{rem:posfep}, 
  the f.e.p. often holds in the presence of weakening. This includes in particular the intuitionistic logic, and it is worth pointing out here 
 Recall that the finite embeddability property for \algcl{HA} 
   was already noted in a classical paper of McKinsey and Tarski \cite[Th. 1.11]{McKinseyT46}.

\section{A Glimpse at Proof Theory}\label{ProofTheory}

The formulas of GBI are all terms constructed from variables $x,y,z,w,x_1,\ldots$ using the
operation symbols $\wedge,\vee,\to,\top,\bot,\cdot,\backslash,/,1$. The set of all these terms is the absolutely free term algebra of this signature, denoted by $\fmal_\text{GBI}$. In this section we give a syntactic sequent calculus that provides a decision procedure for testing if an inequality $s\le t$ holds in all GBI-algebras. The proof that the procedure is complete uses the notion of distributive residuated frame and algebraic cut-elimination, due to Galatos and Jipsen \cite{GalatosJ13,GalatosJ}. The equational decidability of the distributive residuated lattice  reduct is also proved by Kozak \cite{Kozak09}. Of course, especially in this overview we have to recall that the technique of algebraic cut-elimination dates back to Belardinelli, Jipsen and Ono \cite{BelardinelliJO04} (see also \cite[Ch. 7]{GalatosJKO07}, \cite{GalatosO10:apal}).

The approach outlined here works for all subvarieties of \algcl{GBI} that are defined by so-called simple structural rules, which includes \algcl{BI} and many others. In addition to proving decidability, the residuated frame approach also provides a completion that shows any (commutative) distributive lattice-ordered monoid can be embedded in a complete GBI-algebra (BI-algebra).\footnote{Needless to say, the residuated frame approach to proof theory of (G)BI is not the only possible one. As we have already pointed out, there is an intimate connection with a massive body of work on proof theory of relevance logics. See comments and references in the Introduction. We should also mention here that there are numerous more recent references, e.g., cut-free proof calculi for subvarieties of \algcl{BI}\ of Ciabattoni and Ramanayake \cite{CiabattoniR17}.
}

We use an algebraic approach for the sequent calculus, allowing rules like associativity of $\cdot$, $\wedge$ to be handled by a simple normalization of terms. This means that we replace the algebra $\fmal_\text{GBI}$ by a homomorphic image in which terms are identified modulo associativity for $\cdot$ and modulo commutativity and associativity for $\wedge$. This is harmless since each term has only finitely many equivalent syntactic forms modulo these identities. In addition a formula $x$ can be considered as any one of
$$x=x\cdot 1=1\cdot x=x\wedge\top=\top\wedge x=x\vee\bot=\bot\vee x$$ when
attempting to match a sequent rule.
To avoid proliferation of the sequent rules, we also do not distinguish between internal and external symbols of the sequent calculus, but we define a notion of context (bunches) that handles the required constraint automatically.

In the sequent rules (quasiequations) below, the \emph{bunches} $u(x)$ are terms in which the variable $x$ occurs only once, and on the term-tree branch where $x$ occurs, only the symbols $\cdot$ and $\wedge$ are allowed to appear. 

\begin{lemma} \label{lem:soundness}
The rules in Table~\ref{GBIseq} are valid quasiequations of \GBI.
\end{lemma}
\begin{proof}[Sketch]
 In any GBI-algebra, $u(x)$ induces an order-preserving term-function of $x$ under any assignment of the other variables in $u(x)$. \hfill\qed
 \end{proof}

\begin{table}
\begin{center}
$\inferrule{\ }{x\le x}$\qquad
$\inferrule{\ }{u(\bot)\le x}$\qquad
$\inferrule{\ }{x\le \top}$\\[12pt]


$\inferrule{u(x\wedge x)\le y}{u(x)\le y}$ $[\wedge_{idem}]$
\quad
$\inferrule{u(x)\le z}{u(x\wedge y)\le z}$ $[\wedge_l]$
\quad
$\inferrule{u(y)\le z}{u(x\wedge y)\le z}$ $[\wedge_l]$
\quad
$\inferrule{x\le y\quad x\le z}{x\le y\wedge z}$ $[\wedge_r]$\\[12pt]

$\inferrule{u(x)\le z\quad u(y)\le z}{u(x\vee y)\le z}$ $[\vee_l]$
\qquad
$\inferrule{x\le y}{x\le y\vee z}$ $[\vee_r]$
\qquad
$\inferrule{x\le z}{x\le y\vee z}$ $[\vee_r]$\\[12pt]

$\inferrule{x\le y\quad u(z)\le w}{u(x\cdot (y\backslash z))\le w}$$[\backslash_l]$
\quad
$\inferrule{x\cdot y\le z}{y\le x\backslash z}$$[\backslash_r]$
\quad
$\inferrule{x\le y\quad u(z)\le w}{u((z/y)\cdot x)\le w}$$[/_l]$
\quad
$\inferrule{x\cdot y\le z}{x\le z/y}$$[/_r]$\\[12pt]

$\inferrule{x\le y\quad z\le w}{x\cdot z\le y\cdot w}$$[\cdot_{lr}]$
\qquad
$\inferrule{x\le y\quad u(z)\le w}{u(x\wedge(y\to z))\le w}$ $[\to_l]$
\qquad
$\inferrule{x\wedge y\le z}{y\le x\to z}$ $[\to_r]$
\end{center}
\caption{The sequent rules of GBI.\label{GBIseq}}
\end{table}

The effectiveness of these quasiequations stems from the observation that each rule contains the same variables in the premise and in the conclusion, 
and each rule (except $\wedge_{idem}$) eliminates a particular symbol either from the left hand side or the right hand side of the conclusion, as indicated by its name.  
When the rules are applied in a backward proof search, the conclusion is matched to the inequality $s\le t$, and this match determines the assignments to the variables in the premise. Furthermore, the premises of each rule (again, except $\wedge_{idem}$ discussed below) are at most as long as the conclusion (using some standard measure of length of a formula). Hence symbols get eliminated from $s,t$ as the search proceeds, and the inequalities in the premises do not grow in length, so after a finite number of steps the search either terminates with axioms as leaves, or having exhausted all possible applications of the rules the conclusion is that $s\le t$ cannot be proved by this sequent calculus.

In the remainder of this section we outline why this proof procedure yields all valid inequalities of \GBI, and how it extends to cover many of its subvarieties.

The following quasiequation, known as the cut-rule, does not have the subformula property:
$$\inferrule{x\le y\quad u(y)\le z}{u(x)\le z} \text{ [cut]}$$
Adding this rule to the GBI sequent calculus makes it quite easy to show that 
sequent proofs (with cut) can emulate Hilbert system proofs in HGBI, and hence the sequent calculus is complete with respect to the equational theory of \GBI. E.g., the cut rule emulates modus ponens in the form $\inferrule{\top\le y\quad y\le z}{\top\le z}$ and a proof of the axiom $x\to(y\to x)$ is given by
$$
\inferrule{\inferrule{\inferrule{\inferrule{x\le x}{x\wedge y\le x}}{x\le y\to x}}{
x\wedge \top\le y\to x}}{\top\le x\to(y\to x) .}
$$
Nevertheless, the cut-rule cannot be used effectively in a proof search, since the variable $y$ in the premise can be instantiated with any formula, hence the search tree is not finite.
One, however, can show that the GBI-calculus without the cut-rule is able to prove exactly the same inequalities as can be proved with the cut-rule. There are several approaches to proving such a \emph{cut-elimination result}, but with our emphasis on algebraic aspects of (G)BI 
 we choose to outline an algebraic approach based on Galatos and Jipsen \cite{GalatosJ}.

A \emph{binary relation} $N$ from a set $W$ to $W'$ is a map $N:W\to\mathcal P(W')$ and a
\emph{ternary relation} $\circ$ on $W$ is a map $\circ:W\times W\to\mathcal P(W)$. Instead of $z\in N(x)$ we write $xNz$, and for sets $X,Y\subseteq W, Z\subseteq W'$ define 
\begin{itemize}
\item $XNz$ iff $xNz$ for all $x\in X$,
\item $xNZ$ iff $xNz$ for all $z\in Z$,
\item $x\circ y=\circ(x,y)$,
\item $X\circ Y=\bigcup\{x\circ y:x\in X,y\in Y\}$ and
\item $\gamma:\mathcal P(W)\to \mathcal P(W)$ by $\gamma(X)=\{z\in W:\forall y
(XNy\text{ implies }zNy)\}$.
\end{itemize}
Note that $\gamma$ is a closure operator on $W$, i.e., $X\subseteq \gamma(X)=\gamma(\gamma(X))$.

A GBI-\emph{frame} is a structure $\m W=(W,W',N,\circ,E,\backslash\!\!\backslash, /\!\!/,\omt,\Rightarrow,\Leftarrow)$, where $N\subseteq W\times W'$, $\circ,\omt$ are ternary relations on $W$, 
$$\ldd, {\Rightarrow}: W\times W'\to \mathcal P(W)\qquad
\rdd,{\Leftarrow}: W'\times W\to\mathcal P(W)
$$
and the following properties hold:
\begin{description}
\item[\bro A\brc] $\gamma((x\circ y)\circ z)=\gamma(x\circ(y\circ z))$, $\gamma((x\omt y)\omt z)=\gamma(x\omt(y\omt z))$,
\item[\bro E\brc] $\gamma(E\circ x)=\gamma(\{x\})=\gamma(x\circ E)$,
\item[\bro N\brc] $x\circ yNz$ iff $xN(z\rdd y)$ iff $yN(x\ldd z)$
\item[\bro D\brc] $x\omt yNz$ iff $xN(z\Leftarrow y)$ iff $yN(x\Rightarrow z)$
\item[\bro I)\ ] $\gamma(x\omt x)=\gamma(\{x\})$, $\gamma(x\omt y)\subseteq \gamma(\{x\})$ and $\gamma(x\omt y)=\gamma(y\omt x)$.
\end{description}
\noindent
The property (N) is called the \emph{nuclear condition}: it ensures that the closure operator $\gamma$ is a \emph{nucleus}, i.e., satisfies $\gamma(X)\circ \gamma(Y)\subseteq \gamma(X\circ Y)$. The nucleus image of a residuated lattice is again a residuated lattice \cite[Thm 3.34]{GalatosJKO07}, which is important for the upcoming definition of Galois algebra. 
Likewise, property (D) is the \emph{distributive nuclear condition} and ensures that $\gamma(X)\omt \gamma(Y)\subseteq \gamma(X\omt Y)$. Together with (I) it implies that $\gamma(X\omt Y)=\gamma(X)\cap\gamma(Y)$, making the nucleus image a distributive lattice \cite[Lem. 2.1, 2.3]{GalatosJ}.

The \emph{Galois algebra} of $\m W$ is $\m W^+=(\mathcal \gamma[\mathcal P(W)],\cap,\cup_\gamma,\circ_\gamma,1,\backslash,/,\to)$ where 
\begin{itemize}
\item $X\cup_\gamma Y=\gamma(X\cup Y)$, 
\item $X\circ_\gamma Y=\gamma(X\circ Y)$, 
\item $X\backslash Y=\{z\in W:X\circ\{z\}\subseteq Y\}$,
\item $X/Y=\{z\in W:\{z\}\circ Y\subseteq X\}$ and
\item $X\to Y=\{z\in W:X\omt\{z\}\subseteq Y\}$.
\end{itemize}

To become familiar with the concept of a GBI-frame, it is a good exercise to prove the following important result. 

\begin{theorem}[\cite{GalatosJ}]
For any GBI-frame $\m W$ the Galois algebra $\m W^+$ is a complete perfect GBI-algebra. 

Conversely, given any GBI-algebra $\m A$, $(A,A,\le,\cdot,\{1\},\backslash,/,\wedge,\to,\leftarrow)$ is a GBI-frame, where the operation $x\leftarrow y$ is defined as $y\to x$.
\end{theorem}

It is easy to see that when $W=W'=A$ and $N={\le}$ then $(\gamma[\mathcal P(W)],
\cap,\cup_\gamma)$ is the MacNeille completion of (the lattice reduct of) $\m A$, hence we immediately have the following result.

\begin{corollary}
\GBI\ is closed under MacNeille completions.
\end{corollary}

We now outline a proof of algebraic cut elimination. We define a ``free'' GBI-frame $\m W_F$ with the property that any inequality $s\le t$ that is satisfied in the Galois algebra $\m W_F^+$ is provable from the rules of Table~\ref{GBIseq} without using the cut-rule. This definition illustrates that the concept of GBI-frame is
flexible and provides frame-semantics for Gentzen proof systems.

Recall that $\fmal_\text{GBI}$ is the absolutely free term algebra over the signature of GBI-algebras, and let $\m W$ be the homomorphic image such that 
$(W,\cdot,1)$ is a free monoid and $(W,\wedge,\top,\bot)$ is a free commutative
monoid with absorbing element $\bot$. Note that $W$ is itself an algebra with
the signature of a GBI-algebra.

A $\cdot,\wedge$-\emph{unary linear polynomial} 
on $W$ is a term with a single designated variable $x_0$ such that only the 
operations $\cdot$ and $\wedge$ appear on the branch from the root to $x_0$. Let
$U$ be the set of all such terms, and define $W'=U\times W$. We use the notation
$u(\_)$ for a polynomial $u\in U$, so e.g., $\_ \cdot y$ is the polynomial 
defined by $u(x_0)=x_0 \cdot y$. The identity polynomial is denoted by $id$.

Now define the relation $N\subseteq W\times W'$ by
\[
x \mathrel{N} (u,y) \qquad\text{iff}\qquad u(x) \le y\text{ is provable from Table~\ref{GBIseq}.}
\]
Then
\begin{gather*}
x \cdot y \mathrel{N} (u,z)\text{ \ iff \ } u(x \circ y) \le z \text{ \ iff \ }x \mathrel{N} (u(\_ \cdot y), z)\text{ \ iff \ }y \mathrel{N} (u(x \cdot \_), z),
\\
x \wedge y \mathrel{N} (u,z)\text{ \ iff \ } u(x {\wedge} y) \le z\text{ \ iff \ }x \mathrel{N} (u(\_ {\wedge} y), z)\text{ \ iff \ }y \mathrel{N} (u(x {\wedge} \_), z).
\end{gather*}
Hence we define 
\begin{center}
$x\circ y=\{x\cdot y\}$\qquad $E=\{1\}$\qquad $x\omt y=\{x\wedge y\}$\\
$(u,y) \rdd x= \{(u(\_ \cdot x), y)\}$\qquad $x \ldd (u,y)= \{(u(x \cdot \_),y)\}$\\
$(u,y) \Rightarrow x= \{(u(\_ \wedge x), y)\}$\qquad $x \Leftarrow (u,y)= \{(u(x \wedge \_), y)\}$ \\
$\m W_F=(W,W',N,\circ,E,\ldd,\rdd,\omt,\Rightarrow,\Leftarrow)$.
\end{center}
It is straightforward to show that $\m W_F$ satisfies (A), (E), (N), (D), (I), so it is a
GBI-frame.

The following result is at the core of algebraic cut-elimination. For $y\in W'$
we define $y^\triangleleft=\{x\in W : x N y\}$.

\begin{lemma} \label{lem:homo}
Let $h:\m W\to \m W_F^+$ be the unique homomorphism that extends the assignment
$h(x_i)=(id,x_i)^\triangleleft$. Then $t\in h(t)\subseteq (id,t)^\triangleleft$ for all terms $t\in W$.
\end{lemma}
\begin{proof}
This is proved by induction on the structure of $t$. For variables $x_i$ we
have $x_0\in h(x_0)$ since $x_0\le x_0$ is an axioms.

For the induction step, assume $s\in h(s)\subseteq (id,s)^
\triangleleft$ and $t\in h(t)\subseteq (id,t)^\triangleleft$. We only check 
that 
\[
s\to t\in h(s\to t)\subseteq (id,s\to t)^\triangleleft,
\] 
since the remaining cases are similar. 

Since $h$ is a homomorphism, we have
\begin{align*}
h(s\to t) & = h(s)\to h(t) \\
& = \{z\in W:h(s)\omt z\subseteq h(t)\}.
\end{align*}
 From $s\in h(s)$ we deduce that 
$z\in h(s\to t)$ implies 
\begin{align*}
s\omt z & = \{s\wedge z\} \\
& \subseteq  h(t)\subseteq (id,t)^\triangleleft.
\end{align*}
 Therefore $s\wedge z\le t$ is Gentzen provable, hence
$z\le s\to t$ is also provable by $[\to_r]$. We conclude that $z\in(id,s\to t)^\triangleleft$ and thus 
\[
h(s\to t)\subseteq (id,s\to t)^\triangleleft.
\]
Next, let $(u,r)\in W'$ and suppose $h(t)\subseteq (u,r)^\triangleleft$. Since $t\in h(t)$ it follows that $u(t)\le r$ is Gentzen provable. Consider any 
$s'\in h(s)$, whence $s'\le s$ is Gentzen provable. From $[\to_l]$ we see that 
\[ u(s'\wedge (s\to t))\le r \]
 is Gentzen provable. Therefore $s'\wedge (s\to t)\in (u,r)^\triangleleft$.

Since every $\gamma$-closed set is an intersection of sets of the form $(u,r)^\triangleleft$, it follows that $s'\wedge(s\to t)\in h(t)$ for all $s'\in h(s)$.
We conclude that
 \[ h(s)\omt(s\to t)\subseteq h(t), \]
 hence by definition of $\to$ in the Galois algebra it is the case that \newline
 $s\to t\in h(s)\to h(t)=h(s\to t)$. 
\hfill\qed
\end{proof}

\begin{theorem} For any $s,t$  
  the following statements are equivalent.
\begin{upbroman}
\item $\mathsf{GBI}\models s\le t$,
\item $\m W_F^+\models s\le t$,
\item $s\le t$ has a cut-free proof using the rules in Table~\ref{GBIseq}.
\end{upbroman}
\end{theorem}
\begin{proof}
(i) implies (ii) since $\m W_F^+$ is a GBI-algebra. Next,  (ii) 
 implies  $h(s)\subseteq h(t)$,  so 
 $s\in h(s)\subseteq h(t)\subseteq(id,t)^\triangleleft$ (Lemma \ref{lem:homo}). Hence $s N (id,t)$ and therefore (iii) follows from the definition of $N$. Finally, 
(iii) implies (i)  via Lemma \ref{lem:soundness}. 
  \hfill\qed
\end{proof}

\noindent
It takes more work to obtain a decision procedure for  
  well-behaved subvarieties of \algcl{nGBI} (cf. \S~\ref{sec:posdec}). 
 The problem is that
the rule $[\wedge_{idem}]$ could lead to an infinite branch during proof search.
As for intuitionistic logic one can restrict to 3-reduced sequents, but one has to define a suitable measure for the length of a sequent to ensure that the sequents in the premise of a rule do not increase in length. For a detailed
discussion on how to resolve these issues we refer to Galatos and Jipsen \cite{GalatosJ} (see, however, Footnote \ref{ft:ramaerror}). 

\begin{rem} \label{rem:posfep}
The paper in question also proves the finite model property for all subvarieties defined by identities in the language of $\{\cdot,\wedge,\vee,1\}$,
showing that they have a decidable equational theory. Moreover,
if any of these varieties is integral  (i.e., $x\vee 1=1$) then the finite embeddability property holds, hence such varieties have a decidable universal theory.
\end{rem}

\section{(B)BI and Separation Logic} \label{SeparationLogic}


We have not said much so far about the formalism that is largely responsible for the popularity of (B)BI  in theoretical computer science: that is, about \emph{separation logic} (SL). It is a form of Hoare logic for programs involving shared mutable data structures. 
We first recapitulate the basic ideas of general-purpose formalisms in \S~\ref{HoareIntro}, then we discuss specific issues addressed by SL and its cousins in \S~\ref{SepSpecific} and finally get into the details of suitable Hoare-style reasoning in \S~\ref{sec:frame}. 

While we believe this section is reasonably self-contained, the overview---aimed mostly at readers with limited background in program verification---must remain  somewhat sketchy by nature. One can find more information in specialized overviews such as an early one by Reynolds \cite{Reynolds02:lics} or a more recent one by O'Hearn \cite{OHearn12:nato}. 

\subsection{Basic Ideas of Floyd-Hoare Logic(s)} \label{HoareIntro}

\emph{Floyd-Hoare logic}  \cite{Floyd67,Hoare69}, most commonly abbreviated to \emph{Hoare logic}, 
 allows both writing specifications of programs and reasoning about their correctness using simple compositional rules. 
 Its central notion is that of a \emph{partial correctness assertion}  (a.k.a. a \emph{Hoare triple}) 
  of the form $\{P\}C\{Q\}$, where \emph{precondition} $P$ and \emph{postcondition} $Q$ are logical predicates, written in some logical formalisms---it might be an extension of ordinary first-order logic or a variant or extension of (B)BI---and $C$ is a \emph{command}, to be made specific below. 
   A Hoare triple is \emph{valid} if whenever $C$ is executed in a state satisfying $P$, it will terminate in a state satisfying $Q$. 

Even when programs are not allowed to directly manipulate pointers, 
 Hoare logic can be puzzling for a beginner.   Textbook examples show it is easy to get the rules wrong for commands as simple assignment, especially when the assignment formula is allowed to involve the old value of a variable being assigned. Consider $\{\top\}X:=a\{X=a\}$. It might seem a valid triple scheme until one realizes that the expression $a$ can be, for example, $X + 1$.

Let us begin with a typical toy programming language IMP  used both in today's standard monographs \cite{PierceSF,Winskel93} and classical references like Hoare's original paper \cite{Hoare69}.
It involves \emph{assignment}, \emph{sequencing}, \emph{conditionals} and \emph{loops}. 
IMP is not doing any (de)allocation, heap access, concurrency or operations on other \emph{shared mutable data structures}; we will turn our attention to these below. 
An execution of an IMP program consists in manipulating  (\emph{global}) \emph{program ``variables''}\footnote{\label{ft:variables} A logician may object whether the word \emph{variables} is really appropriate here. Sometimes the term \emph{(storage) locations} is used instead \cite{Winskel93}, but as the reader will recall, we already used this  name in \S~\ref{HeapModel} for \emph{pointer} labels and will continue to do so below.} $\mathit{PVar} = \{X_0, X_1, X_2, \dots\}$ by assigning to them arithmetical expressions built using basic arithmetical operations (addition, optionally also multiplication and/or truncated subtraction) from $\mathit{PVar}$ and numerical constants. Conditionals and loops are guarded by boolean expressions $b_0$, $b_1, \dots$ which are built using standard boolean connectives from atoms comparing arithmetical expressions for (in)equality. Finally, in the language of assertions (but not  IMP itself!) we also allow class of another genuine,  \emph{quantified assertion variables} $\mathit{AVar} = \{v_0, v_1, v_2, \dots\},$ for which one can substitute arithmetical expressions. 
 This is enough to characterize all commands of IMP by axioms in Table \ref{fig:HoareImp}.

\begin{table}
\begin{center}
$\inferrule{}{\{P\}\cf{SKIP}\{P\}}$ \qquad $\inferrule{}{\{P[a/X]\}X:=a\{P\}}$  \qquad $\inferrule{\{P\}C_1\{Q\} \\ \{Q\}C_2\{R\}}{\{P\}C_1;C_2\{R\}}$ 

\medskip

$\inferrule{\{P \wedge b\}C_1\{Q\} \\ \{P \wedge \neg b\}C_2\{Q\}}{\{P\}\cf{IF } b \cf{ THEN } C_1 \cf{ ELSE } C_2\{Q\}}$ \quad $\inferrule{\{P \wedge b\}C\{P\}}{\{P\}\cf{WHILE }b\cf{ DO }C\cf{ OD}\{P \wedge \neg b\}}$.

\end{center}
\caption{\label{fig:HoareImp} Axioms and rules for commands of IMP.}

\end{table}

 
 In addition, manipulation of the Hoare calculus requires rules that are, interestingly, often called \emph{structural}\ \cite{Kleymann99:fac,OHearnRY01,Reynolds02:lics,Calcagno2007,OHearn12:nato}. Here, this word is taken in a somewhat different meaning than the one known to proof theorists. Namely, it denotes the rules which allow modifying pre- and post-conditions, as opposed to specifying complex program expressions in terms of their constituent subprograms  along the lines of Table \ref{fig:HoareImp}. Nevertheless, as we are going to see in \S~\ref{sec:frame}, the central rule of separation logic connects this meaning of ``structurality'' with the one familiar to substructural logicians! 

  Perhaps the most well-known ``structural'' rule is \emph{consequence}:

\begin{center}
$\inferrule*[right=Conseq]{\vDash P' \to P \\  \{P\}C\{Q\} \\ \vDash Q \to Q'}{\{P'\}C\{Q'\}.}$
\end{center}
\noindent
Recall that the language of assertions includes quantified arithmetical statements. Hence, we use the semantic theoremhood $\vDash$ rather than the syntactic theoremhood $\vdash$ in the statement of this rule: there are obvious G\"odelian limitations meaning there cannot be any complete yet recursively axiomatizable notion of proof (see however \S~\ref{sec:DecRev} below). 
Obviously, this  entails that
 in practical applications one can only look for decision procedures for well-behaved fragments. Such limitative results also open up application areas for partially automated proof assistants as an alternative to fully automated tools.

\begin{rem} \label{rem:semantics}
Semantics, either \emph{operational} or \emph{denotational}  \cite{Winskel93}, can be given using the notion of \emph{store} (stack) as introduced in \S~\ref{HeapModel}.\footnote{\label{ft:partial} Note that the assertion for \cf{WHILE} is valid only when read as a \emph{partial} one. 
 That is,  $\{P\}C\{Q\}$ is read as \emph{\textbf{if} $C$ is started in a store satisfying $P$ \textbf{and} terminates, \textbf{then} any store it terminates in satisfies $Q$}. Under this reading, for example, $\{\top\}C\{\bot\}$ simply specifies that $C$ never terminates, regardless of the original values of program variables. An alternative reading, usually denoted as $[P]C[Q]$, is the \emph{total} one:  \emph{\textbf{if} $C$ is started in a store satisfying $P$, \textbf{then} it  terminates \textbf{and} any store it terminates in satisfies $Q$}.}
However, some proponents of Hoare-style formalisms point out that the above set of rules can be seen as semantics in its own right; hence the name \emph{axiomatic semantics} (cf., e.g., \cite[Ch. 6]{Winskel93}). One can prove soundness theorems connecting it to more standard semantics. There is even a form of completeness available, although one has to tread carefully here due to the G\"odelian limitations mentioned above: 
 namely, these rules allow deriving suitable \emph{weakest preconditions} and consequently all valid IMP-triples \emph{in the presence of an oracle for elementary arithmetic}. This is called \emph{relative completeness} \cite{Cook78} (see also \cite{Clarke85}, \cite[Ch. 7]{Winskel93}). We will continue the discussion of semantical aspects of correctness assertions in Example \ref{ex:frameside} and Remark \ref{rem:semframe} below.  
\end{rem}

\noindent
The popularity of Hoare logics, however, does not stem so much from theoretical results like relative completeness (available only in a restricted context anyway \cite{Clarke85}), but from its applications to program specification and verification. 
 For example, they allow (semi-)automated verification of programs \emph{annotated}/\emph{decorated} with assertions via extraction of \emph{verification conditions} \cite[\S~7.4]{Winskel93}, \cite{PierceSF}. There are programming languages with specification constructs built-in, like Eiffel \cite{Meyer97} or Dafny \cite{Leino2010}, but for scalable analysis of industrial-size code in a general-purpose language one uses  analysis platforms like Frama-C \cite{Cuoq2012},  allowing annotations written in an external specification language (e.g., ACSL).

So much for bird's eye view of general-purpose (Floyd-)Hoare logic(s). Now where and how does the connection with (B)BI enter the picture?

\subsection{Heap(let)s, Allocation and Separation} \label{SepSpecific}

As we have already indicated, specification and automated verification become particularly problematic in the presence of shared mutable data structures. 
 O'Hearn, Reynolds and Yang \cite{OHearnRY01} summarized this as follows:

\begin{quote}
The main difficulty is not one of finding an in-principle adequate axiomatization of pointer operations; rather there is a mismatch between simple intuitions about the way that pointer operations work and the complexity of their axiomatic treatments. For example, pointer assignment is operationally simple, but when there is aliasing, arising from several pointers referring to a given cell, then an alteration to that cell may affect the values of many syntactically unrelated expressions. 
\end{quote}
The idea that substructural connectives can help axiomatic approaches to the assertion language can be traced back  to Burstall  \cite{Burstall72:mi}. 
Much later, Reynolds \cite{Reynolds00:intuitionistic,Reynolds02:lics} and Ishtiaq and O'Hearn \cite{IshtiaqOH01:popl} turned the idea into a working, well-defined Hoare-style language. These papers also provide
 more references to earlier attempts at a suitable verification logic (cf. \S~\ref{sec:competing}). 

 To understand  the advantages of assertions expressed in (an extension of) (B)BI, let us continue the above quote from O'Hearn et al. \cite{OHearnRY01}: 

\begin{quote}
We suggest that the source of this mismatch is the global view of state taken in most formalisms for reasoning about pointers.  \dots To understand how a program works, it should be possible for reasoning and specification to be confined to the cells that the program actually accesses. The value of any other cell will automatically remain unchanged. 
\end{quote}
 \noindent
 Thus, substructural connectives are used to express and combine assertions about  disjoint portions of heap: an assertion  talks about \dots

\begin{quote}
 \dots\ a heaplet rather than the global heap, and a spec $\{P\} C \{Q\}$ says that if $C$ is given a heaplet satisfying $P$ then it will never try to access heap outside of $P$ (other than cells allocated during execution) and it will deliver a heaplet satisfying $Q$ if it terminates. \cite{Berdine2006}
\end{quote}
It is easy to guess now that assertions about disjoint heaplets are combined using the fusion $*$ of commutative (B)BI, which in this community has alternative names like the \emph{spatial conjunction}, the \emph{separating conjunction} or the \emph{independent conjunction}.  
Its residual $\bi$ is commonly referred to as the \emph{magic wand}  or \emph{separating implication}: 
 $P \bi Q$ means \emph{whenever the present heaplet is extended with a disjoint heaplet satisfying $P$, the resulting heaplet satisfies $Q$}. The use of $\bi$ in Hoare triples of separation logic (SL) seems to have been proposed first by  Ishtiaq and O'Hearn \cite{IshtiaqOH01:popl}.   Just like its additive counterpart, 
   $\bi$ is particularly useful when specifying and deriving \emph{weakest preconditions}.\footnote{On the other hand, the semantic clause of $\bi$ quantifies over the collection of all possible disjoint extensions satisfying $Q$, which can be 
 infinite, and is problematic from a model checking point of view. For this reason, there is a line of research  dealing with \emph{adjunct elimination} for separation logic and related formalisms  \cite{Lozes04:entcs,Dawar04,Calcagno07:popl,CalcagnoDY10:ic}.}
 
 For a concrete example of a suitable programming language equipped with a Hoare logic, let us take our inspiration from Reynolds \cite{Reynolds02:lics}. And for a semantic intuition, let us  return to models of (B)BI discussed in \S~\ref{HeapModel}, especially the \emph{stack-and-heap model} \cite[\S~2.2]{BrotherstonK14:jacm}  introduced at the end of that subsection. As we pointed out therein, the name \emph{store} also used by, e.g., Demri and Deters \cite{DemriD15:jancl} would be perhaps more adequate, hence we speak about  the \emph{store-and-heap model} instead. In order to allow full pointer arithmetic, let us identify locations,  record values and store values, i.e., (in the notation of  \S~\ref{HeapModel}) take $L = RV = \mathit{Val} = \nat$.  The advantage of operational semantics based on such a model (i.e., on a set-theoretic product of the collection of stores and the collection of heaps) is that it allows extending IMP 
   with dynamic commands 
  in a fully orthogonal way. Of course, as we incorporate IMP with its assertion language, G\"odel's Incompleteness Theorem still applies; 
 in \S~\ref{sec:DecRev} below, we discuss restrictions allowing more positive results. 

We take thus 
 IMP from \S~\ref{HoareIntro} and add 
 primitives for
 \begin{center}
  \emph{allocation} $X := \cf{CONS}(a_0, \dots, a_{n-1})$, \emph{lookup} $X := [a]$, \emph{mutation} $[a] := a'$, and \emph{deallocation} $\cf{DISPOSE}\ a$. 
  \end{center}
   We also take the entire 
 assertion language introduced in  \S~\ref{HoareIntro}  and allow the use of $*$ and $\bi$ to form new assertions. Furthermore, we extend the language of assertions with new atomic constructs:
\begin{itemize}
\item a constant $emp$ true at any pair $(s,\emptyset)$ where $\emptyset$ is the empty heap(let) and
\item a family of  \emph{pointer atoms} $a_1 \mapsto a_2$ which hold at those pairs $(s,h)$ where $h$ is a singleton heap(let), i.e., defined on exactly one location, which happens to be $\hat{s}(a_1)$ ($\hat{s}$ denotes the inductive extension of $s$ to arbitrary arithmetical expressions) and $h(\hat{s}(a_1)) = \hat{s}(a_2)$.
\end{itemize}

 The presence of $*$ in the language means we can describe any concrete finite heap using expressions of the form 
$
(a_1 \mapsto e_1) * \dots * (a_n \mapsto e_n).
$ 
 Given $\overline{a} := a_0, \dots, a_{n-1}$, let us also introduce an abbreviation for $e$ pointing to the head of a dynamic list storing numbers denoted by $\overline{a}$:
\[
 e \mapsto_{\ell} \overline{a} := (e \mapsto a_0 ) * (e +1 \mapsto a_1) * \dots * (e + (n  -1)\mapsto a_{n-1}).
 \]
 Reynolds \cite{Reynolds02:lics} or Ishtiaq and O'Hearn \cite{IshtiaqOH01:popl} point out that the validity of $\{P\}C\{Q\}$ (whether as \emph{partial} or \emph{total} correctness assertions, cf. Rem. \ref{rem:semantics} and especially Footnote \ref{ft:partial}) 
  should entail not only that whenever $P$ holds at a pre-execution state $(s,h)$,  then $Q$ would  hold after the execution of $C$, but also that $C$ executed at $(s,h)$ is \emph{safe}, i.e., does not lead to a memory fault. That is, it should never try to mutate, lookup or dispose of a cell which has not been previously allocated.\footnote{When specifying this property in a proof assistant, one can do it in a (co)inductive fashion. In a metatheory allowing excluded middle at least for assertions, one can work with two inductive properties \texttt{fault} and \texttt{no\_fault} and show (using excluded middle) that they are complementary, i.e., that for any $(s,h)$  and $C$ exactly one of the two holds. Alternatively, one can stay within constructive metatheory by making  \texttt{no\_fault} coinductive. 
 See \cite{Paulus2016} for an example of a student-oriented formalization in a proof assistant discussing these issues.} 
 
 Under this reading, we can salvage all of the rules for IMP discussed in \S~\ref{HoareIntro}, though some of the clauses (especially the one for $\cf{WHILE}$) require more work due to the safeness requirement. Moreover, we can finally give the \emph{axiomatic semantics} of new, dynamic constructs in Table \ref{fig:smallax}.
  Note here how we use the auxiliary, quantifiable variables of the metalang\-uage to keep original values of $X$ in the allocation and lookup clauses. Recall again from \S~\ref{HoareIntro} that in distinguishing program ``variables'' $PVar$ from  variables of  the metalanguage $AVar$ we follow standard references like the Winskel book \cite{Winskel93} rather than, e.g., the presentation of Reynolds \cite{Reynolds02:lics}. A similar approach to ours is also taken by Demri and Deters \cite{DemriD15:jancl}.
  
 \begin{table} 
$$\inferrule{}{\{ (X = v) \wedge emp\} \, X := \cf{CONS}(\overline{a}) \, \{ X \mapsto_\ell (\overline{a}[v/X])\}},$$ 

$$\inferrule{}{\{(X = v) \wedge (a \mapsto v')\} \, X := [a] \, \{ (X = v') \wedge a[v/X] \mapsto v'\}},$$

$$\inferrule{}{\{\exists v. a \mapsto v\} \, [a] := a' \, \{a \mapsto a'\}},$$  

$$\inferrule{}{\{\exists v. a \mapsto v\} \, \cf{DISPOSE}\ a\, \{ emp \}}.$$  

\caption{\label{fig:smallax} \emph{Small} or \emph{local} axioms for dynamic commands.}

\end{table}

\subsection{Local Axioms, Global Specifications and The Frame Rule} \label{sec:frame}

The axioms in Table \ref{fig:smallax} are  \emph{local} or  \emph{small} (cf., e.g., \cite{OHearn12:nato}). 
 A specification for $C$ is local if only involves variables used by $C$ and what O'Hearn called the \emph{footprint} of $C$: parts of the heap used during its execution. 
 While it is easy to see why the modular approach postulated at the beginning of this section requires such small axioms, it also calls for suitable \emph{structural} rules in the sense already mentioned in \S~\ref{HoareIntro}.  
 We need a rule which allows deriving triples of the form $\{P*R\}C\{Q*R\}$ from triples of the form $\{P\}C\{Q\}$. In other words, the central \emph{structural} rule of separation logic must be very much a structural rule in ordinary proof-theoretic sense: an introduction rule for  $*$. 

Nevertheless, without suitable restrictions, such an inference can be unsound if  $C$ involves an allocation, lookup or assignment to some $X$ (i.e., with $X$ on the left side of ``$:=$'') appearing in $R$. 
  
  \begin{expl} \label{ex:frameside} 
  Consider, for example, $C$ being $X := \cf{CONS}(2)$, $P$ and $Q$ being the constantly true assertion $\top$ and $R$ being the assertion $X = 2$. We do have that $\vDash \{ \top \} X := \cf{CONS}(2) \{ \top \}$,\footnote{In our ideal mathematical world, where heaps can be arbitrarily large as long as they are finite, allocation never leads to a memory fault (unlike lookup, mutation and deallocation).} but 
\begin{equation} \label{eq:consfails}
 \nvDash \{ \top * (X = 2) \} X := \cf{CONS}(2) \{ \top * (X = 2) \}.
 \end{equation} 
 
 \noindent
 To understand \refeq{eq:consfails}, recall that the execution of $X := \cf{CONS}(\overline{a})$ in $(s,h)$ can transition to any $(s',h')$, where $h'$ is obtained by extending $h$ with a contiguous interval of fresh heap addresses pointing at the (values denoted at $s$ by) subsequent elements of $\overline{a}$, and $s'$ is just $s$ modified at $X$ to store the newly allocated address of the first element of the list. In the case of a singleton list $\overline{a} = [2]$, we can start with $h$ being empty and $h'$ consisting, e.g., of a single pair $(1, 2)$, meaning that $s'(X) = 1$, even if $s(X) = 2$. We thus obtain an example of a (non-faulting) execution starting in a state satisfying the precondition and resulting in a state where the postcondition fails---a counterexample to the validity of the triple.
 \end{expl}
 Define $appear(R)$ as the set of program variables syntactically appearing\footnote{Of course, only \emph{free} occurrences matters. But in our language there are no binders for elements of $PVar$ (unlike $AVar$).} in the assertion $R$ and the predicate \emph{modifies} as shown in Table \ref{fig:modi}. The \emph{frame rule}\footnote{\label{ft:frameorigin} 
  The term \emph{frame problem} was originally proposed in a classical  1969 paper \cite{McCHay69} on problems of knowledge representation in artificial intelligence. \takeout{Kassios \cite{Kassios11} restates the general formulation: \emph{when formally describing a change in a system, how do we specify what parts of the state of the system are not affected by that change?}}  The realization that such problems 
   arise also in formal specifications using Floyd-Hoare logics predates the development of separation logic. An example, focusing on the issues of object-oriented specifications with inheritance,  is provided by a 1995 paper by Borgida et al. \cite{BorgidaMR95}.} proposed by O'Hearn \cite{OHearnRY01,IshtiaqOH01:popl} is hence
$$
\inferrule*[right=Frame]{\{P\} \, C \, \{Q\} \\ modifies(C) \cap appears(R) = \emptyset}{\{P*R\}\,C\,\{Q*R\}}.
$$

 \begin{table}
\begin{align*}
 \mathit{modifies}(X := a) & =  \mathit{modifies}(X := [a]) \\ & =  \mathit{modifies}(X := \cf{CONS}(\overline{a})) = \{X\}, \\
 \mathit{modifies}(\cf{SKIP}) & =  \mathit{modifies}(\cf{DISPOSE}\ a) \\ & = \mathit{modifies}([a] := a') = \emptyset, \\
 \mathit{modifes}(C_1;C_2) & =  \mathit{modifies}(\cf{IF } b \cf{ THEN } C_1 \cf{ ELSE } C_2) \\ & = \mathit{modifies}(C_1) \cup \mathit{modifies}(C_2), \\
 \mathit{modifies}(\cf{WHILE }b\cf{ DO }C\cf{ OD}) & =  \mathit{modifies}(C).
 \end{align*}
\caption{\label{fig:modi} Predicate \emph{modifies}.}
\end{table}

\begin{rem} \label{rem:semframe}
For a logician and perhaps even more so for an algebraist, the presence of side conditions such as $modifies(C) \cap appears(R) = \emptyset$ is certainly disappointing.  
 This sentiment is shared by theoretical computer scientists: 

\begin{quote}
Hoare logic is bedevilled by complex but coarse side conditions
on the use of variables. \cite{ParkinsonBC06:lics} 
\end{quote}

\noindent
Some readers may be puzzled by the fact that in some references (e.g., \cite{Calcagno2007}) the frame rule is stated \emph{without} side conditions nonetheless. 
 Calcagno et al. \cite[\S~1]{Calcagno2007} claim that such  conditions can be avoided thanks to the absence of ``the traditional Hoare
logic punning of program variables as logical variables'', crediting \cite{BornatCY06:mfps,ParkinsonBC06:lics} with the idea. 

The quote might be somewhat confusing, depending on what \emph{punning of program variables as logical variables}  is taken to mean. 
Recall once again that in our setting we \emph{do} distinguish between program variables (\emph{storage locations}, cf. Footnote~\ref{ft:variables}) $PVar = \{X_1, X_2, X_3 \dots\}$ and quantified \emph{assertion variables} $AVar = \{v_1, v_2, v_3 \dots\}$.  
  More informative descriptions of the problem are 
  \begin{quote}
[Program] variables ought
to be resource, treated formally by the logic and not mumbled over in side
conditions. \cite{BornatCY06:mfps}
\end{quote}
and, still more precisely,
\begin{quote}
Hoare logic
does not allow us to describe the ownership of [program] variables \dots Separation logic divides the store into stack---the variables
used by a program---and heap---dynamically allocated
records---but does not give any formal treatment of the stack. \cite{ParkinsonBC06:lics}
\end{quote}
Returning with this insight to Example \ref{ex:frameside}, we can see that the use of $*$ in pre- and post-conditions is entirely irrelevant in \refeq{eq:consfails}: $X$ lives in the store (``stack''), not on the heap, and any (in)equality statement about its value at a given $(s,h)$ will be also valid at $(s,h')$, for any other $h'$. In other words, equality judgements, including those involving members of $PVar$, ``spread beyond the separating conjunction''; they are \emph{heap-independent}. The syntactic apparatus of SL is indeed not tailored to improve control of the store. Hence, avoiding problematic side conditions not only departs ``from the theoretical tradition in program logic'' \cite{Calcagno2007}, but necessitates restricting/complicating the assertion language and the programming language in question. The setup of  Calcagno et al. \cite{Calcagno2007}, for example, does not cover alteration of programming variables, whereas that of Parkinson et al. \cite{ParkinsonBC06:lics} not only relies on explicit ``ownership'' predicates, but also on side conditions more familiar in algebra and logic, i.e., standard freshness assumptions. Still, such approaches are particularly relevant in applications of separation logic focusing on concurrency rather than pointer reasoning (cf. \S~\ref{sec:concur}).
\end{rem}

\noindent
With the frame rule at our disposal, we can derive \emph{global} specifications from local ones---and in the process use basic axioms and rules of (B)BI. We are providing these inferences explicitly in Table \ref{fig:derglo} below; to save some space, we  abbreviate ``$\exists v.~a~\mapsto~v$'' as ``$a \mapsto \_$''. 
 Such derivations are outlined, e.g., by Reynolds \cite{Reynolds02:lics}, 
 Yang \cite{YangPhD} or O'Hearn \cite{OHearn12:nato}.  Let us discuss their most salient points. 
 
\ifbook\else
  Deallocation is rather straightforward: 
 $
  \{( a \mapsto \_) * R \} \, \cf{DISPOSE}\ a\, \{ R\}$.
Its derivation in Table \ref{fig:derglo} 
uses the fact that $emp$ is the monoidal unit. 
 One arrives almost instantly at the form suitable for \emph{backward reasoning} \cite{IshtiaqOH01:popl}, i.e., 
allowing  \emph{backwards} program annotations, starting from an  arbitrary postcondition $R$. 
  Such reasoning is at the heart of many applications of Hoare-style formalisms, in particular  derivations of \emph{weakest preconditions}, which in turn are essential for relative completeness results (cf. Remark \ref{rem:semantics}). 
 A discussion of such results in the context of separation logic was provided in an early stage of its development  by Yang 
   \cite{YangPhD}. 
\fi

\begin{landscape}
\begin{table}
 \centering
 \[
\textbf{Deallocation: \qquad\qquad} \inferrule*[right=Conseq]{
	\inferrule*[left=Frame]{
		\{ a \mapsto \_ \} \, \cf{DISPOSE}\ a\, \{ emp \}
	}{\{( a \mapsto \_) * R \} \, \cf{DISPOSE}\ a\, \{ emp * R\}} \\ \inferrule*{}{\vDash (emp * \alpha) \to \alpha}
	}{\{(a \mapsto \_) * R \} \, \cf{DISPOSE}\ a\, \{ R\}.}
\] 

\[ 
\textbf{Mutation: \qquad\qquad} \inferrule*[right=Conseq]{
	\inferrule*[left=Frame]{
		\inferrule*{}{
			\{a \mapsto \_\} \, [a] := a' \, \{a \mapsto a'\}}
		}{\{\{a \mapsto \_)*((a \mapsto a') \bi R)\} \, [a] := a' \, \{(a \mapsto a')*((a \mapsto a') \bi R)\}} \\ \vDash (\alpha*(\alpha \bi \beta)) \to \beta
	}{\{\{a \mapsto \_)*((a \mapsto a') \bi R)\} \, [a] := a' \, \{R\}.}
  \] 
  
 \[
\textbf{Lookup: } \inferrule*[Right=VarEl]{ 
	\inferrule*[Right=Conseq]{
		\inferrule*[Right=subst]{
 			\inferrule*[Right=Conseq]{
 				\inferrule*[Left=Frame]{
					\inferrule*{}{
						\{(X = v) \wedge (a \mapsto v')\} \, X := [a] \, \{ (X = v') \wedge a[v/X] \mapsto v'\}}
					}{\{((X = v) \wedge (a \mapsto v')) * \alpha[v'/X]\} \, X := [a] \, \{ ((X = v') \wedge a[v/X] \mapsto v') * \alpha[v'/X] \} } \, \vDash ((X = v') \wedge \beta) * \gamma[v'/X] \to \beta * \gamma}
		 		{\{((X = v) \wedge (a \mapsto v')) * \alpha[v'/X]\} \, X := [a] \, \{ (a[v/X] \mapsto v') * \alpha \}}
			}{\{((X = v) \wedge (a \mapsto v')) * ((a[v/X] \mapsto v') \bi R[v'/X])\} \, X := [a] \, \{ (a[v/X] \mapsto v') * ((a[v/X] \mapsto v') \bi R) \} \\\\ \vDash  (X = v) \wedge (\beta * \gamma) \to ((X = v) \wedge \beta) * \gamma[v/X]  \\ \vDash (\alpha*(\alpha \bi \beta)) \to \beta}
		}{\{(X = v) \wedge ((a \mapsto v') * ((a \mapsto v') \bi R[v'/X]))\} \, X := [a] \, \{ R \}}
	}{\{(a \mapsto v') * ((a \mapsto v') \bi R[v'/X])\} \, X := [a] \, \{ R \}}
 \]
 
 \[  \textbf{Allocation:} 
 \inferrule*[Right=conseq]{
 	\inferrule*[Right=conseq]{
 		\inferrule*[Right=subst]{
			 \inferrule*[Left=conseq]{
 				\inferrule*[Left=conseq]{
			 		\inferrule*[Left=Frame]{
					 		\{ (X = v) \wedge emp\} \, X := \cf{CONS}(\overline{a}) \, \{ X \mapsto_\ell (\overline{a}[v/X])\}}
						{\{ ((X = v) \wedge emp) * \alpha[v/X]\} \, X := \cf{CONS}(\overline{a}) \, \{ (X \mapsto_\ell (\overline{a}[v/X])) * \alpha[v/X]\}} \\\\ \vDash (X = v) \wedge \alpha \to ((X = v) \wedge emp) * \alpha[v/X] \\ \vDash \beta \to \exists v'. (X = v') \wedge \beta[v'/X]}
					{\{  (X = v) \wedge \alpha \} \, X := \cf{CONS}(\overline{a}) \, \{ \exists v'. (X = v') \wedge ((v' \mapsto_\ell (\overline{a}[v/X])) * \alpha[v/X])\}} \\ \vDash (\forall v'. \alpha) \to \alpha}
				{\{ \forall v'.  (X = v) \wedge \alpha \} \, X := \cf{CONS}(\overline{a}) \, \{ \exists v'. (X = v') \wedge ((v' \mapsto_\ell (\overline{a}[v/X])) * \alpha[v/X])\}}
			}{\{ \forall v'.  (X = v) \wedge ((v' \mapsto_\ell (\overline{a}[v/X])) \bi R[v'/X]) \} \, X := \cf{CONS}(\overline{a}) \, \{ \exists v'. (X = v') \wedge ((v' \mapsto_\ell (\overline{a}[v/X])) * ((v' \mapsto_\ell (\overline{a}[v/X])) \bi R[v'/X]))\} \\\\ \vDash (X = v) \wedge \alpha[v/X] \to \alpha \\ \vDash (\alpha*(\alpha \bi \beta)) \to \beta}
		}{\{ \forall v'.  (v' \mapsto_\ell \overline{a}) \bi R[v'/X]) \} \, X := \cf{CONS}(\overline{a}) \, \{ \exists v'. (X = v') \wedge  R[v'/X]\}} \\\\ \vDash (\exists v'. (X = v') \wedge \alpha[v'/X]) \to \alpha} 
	{\{ \forall v'.  (v' \mapsto_\ell \overline{a}) \bi R[v'/X]) \} \, X := \cf{CONS}(\overline{a}) \, \{ R \}} 
 \]
 
 \medskip
 
 \caption{\label{fig:derglo} Derivations of \emph{global backwards} specifications for local (small) axioms.}
\end{table}
 \end{landscape} 

\ifbook
Deallocation is rather straightforward: 
 $
  \{( a \mapsto \_) * R \} \, \cf{DISPOSE}\ a\, \{ R\}$.
Its derivation in Table \ref{fig:derglo} 
uses the fact that $emp$ is the monoidal unit. 
 It arrives almost instantly at the form suitable for \emph{backward reasoning} \cite{IshtiaqOH01:popl}, i.e., 
allowing  \emph{backwards} program annotations, starting from an  arbitrary postcondition $R$. 
  Such reasoning is at the heart of many applications of Hoare-style formalisms, in particular  derivations of \emph{weakest preconditions}, which in turn are essential for relative completeness results (cf. Remark \ref{rem:semantics}). 
 A discussion of such results in the context of separation logic was provided in an early stage of its development  by Yang 
   \cite{YangPhD}. 
\else\fi

The global backward specification for mutation:
\[ 
\{(a \mapsto \_)*((a \mapsto a') \bi R)\} \, [a] := a' \, \{R\},
\]
while still very simple to derive, is the first one where we need residuation: 
 \begin{equation} \label{eqres}
 \vDash (P*(P \bi Q)) \to Q
 \end{equation} 
 (note that in Table \ref{fig:derglo} we often use Greek letters as metavariables ranging over assertions, if we want to instantiate them in the next step; we sometimes denotes this act of substitution as  \textsc{subst}). 
  
 The global backward specification for lookup:
 \[
 \{(a \mapsto v') * ((a \mapsto v') \bi R[v'/X])\} \, X := [a] \, \{ R \}
 \]
  requires a bit more effort. Apart from using again 
   \refeq{eqres}, instances of the consequence rule used in the derivation also use laws governing interactions of $PVar$, $AVar$, 
    lattice and substructural connectives and heap-independent assertions, such as the equivalence:
\[
\vDash  (X = v) \wedge (P * Q) \leftrightarrow ((X = v) \wedge P) * Q[v/X].
\]
Moreover, we also need an application of the (derivable or admissible) rule
\[
\inferrule*[Right=VarEl]{\{(X = v) \wedge P \} \, C  \, \{ R \} \\ v \text{ fresh for } P, R}{\{ P \} \, C  \, \{ R \}}.
\] 
 In references like Reynolds \cite{Reynolds02:lics} or Yang \cite{YangPhD}, there are special rules like \emph{auxiliary variable renaming} and \emph{auxiliary variable elimination} which can be used to derive such rules. In \S~\ref{sec:proofsl}, we present another proof system where this rule is indeed derivable rather than primitive.
 
Similarly, obtaining the global backward axiom for allocation
 
 \[
 \{ \forall v'.  (v' \mapsto_\ell \overline{a}) \bi R[v'/X]) \} \, X := \cf{CONS}(\overline{a}) \, \{ R \}
 \]
 requires  using the frame and consequence rules jointly with BI laws and basic quantification laws, in particular  $\vDash (\exists v'. (X = v') \wedge R[v'/X]) \to R$.

 \section{Proof Theory and Decidability for Fragments of SL} \label{sec:sltheory}

We have argued that 
 proof theory of separation logic can be seen as an extension of proof theory of BI and substructural logics. It would be misleading, however, to give the impression that the only potential r\^ole of proof theory lies in deriving  general axioms  like those discussed above. If that were so, the reader may ask, why not simply begin with postulating the axioms in a suitable ``global backwards'' form? Furthermore, such a critical reader may be perplexed by questions of decidability, both in the light of the discussion in \S~ \ref{HoareIntro} and the one in \S~\ref{Decidability}. In this section, we are going to address both issues.

\subsection{Sketch of a Proof System for SL} \label{sec:proofsl}

As we stated in \S~\ref{HoareIntro}, the industrial importance of Floyd-Hoare logics  indeed does not quite stem from relative completeness results via calculation of schemes of weakest preconditions. While such results are an important theoretical  characterization, the real practical interest lies in deriving and verifying annotations and specifications of concrete programs. 

\newcommand{\deds}{\,\Rrightarrow}

Let us then take stock 
 recapitulating which axioms and rules were exactly used  in Table \ref{fig:derglo}. The resulting proof system can derive not only these ``global  backwards specifications'', but---as the reader can verify---meaningful pieces of annotated code, similar to those used as examples, e.g., in \cite{Reynolds02:lics}. 
   We propose that judgements $\vdash \{P\} \, C \, \{Q\}$ are deduced using the following axioms and rules:

\begin{itemize}
\item axioms in Tables \ref{fig:HoareImp} and \ref{fig:smallax}; 
\item the \textsc{Frame} rule;
\item rules
$$
\inferrule*[right=VarEl$\exists$]{\vdash\{ P \} \, C  \, \{ R \} \\ v \text{ fresh for }  R}{\vdash\{\exists v. P \} \, C  \, \{ R \}} \quad 
$$
and
$$
\inferrule*[right=VarEl$\forall$;]{\vdash\{ P \} \, C  \, \{ R \} \\ v \text{ fresh for }  P}{\vdash\{ P \} \, C  \, \{ \forall v.R \}}
$$
\item a \emph{deductive} version of the consequence rule
\begin{center}
$\inferrule*[right=ConDed]{P' \deds P \\ \vdash \{P\}C\{Q\} \\  Q \deds Q'}{\vdash \{P'\}C\{Q'\},}$
\end{center}
where 
 $P' \deds P$ and $Q \deds Q'$ are derived using:
\begin{itemize}
\item axioms and rules of any proof system which is equipollent with (can derive all theorems of) the Hilbert-style system for BI presented in \S~\ref{LogicAlgebra}. To fix attention, let us take  the system presented in \S~\ref{ProofTheory} plus commutativity (with $\leq$ replaced by $\deds$). In order to keep as close as possible to \S~\ref{ProofTheory}, we use here the notation $\alpha(\varphi)$  where $\alpha$ denotes a \emph{bunch} from \S~\ref{ProofTheory} adjusted to the present syntax;  
\item the basic theory of equality (cf. \cite[\S~4.7]{Troelstra1996}), i.e., 
\[
\inferrule{\alpha(a = a) \deds \varphi}{\alpha(\top) \deds \varphi} 
\]

and
$$
\inferrule{\alpha(a_1 = a_2,  \psi[a_1/v], \varphi[a_2/v]) \deds \chi}{\alpha(a_1 = a_2, a_1 = a_2 \wedge \psi[a_1/v], a_1 = a_2 \wedge \varphi[a_2/v]) \deds \chi}
$$
where $\alpha(x,y,z)$ is the ternary counterpart of the notion of a \emph{bunch}\footnote{Note that this rule allows to spread equality statements across the bunch. In the store-and-heap semantics, such atoms are \emph{heap-independent}: they only depend on the store.}  from \S~\ref{ProofTheory}, i.e., a scheme of an assertion formula  in which each of the schematic variables $x$, $y$ and $z$ occurs only once, and on the term-tree branches where $x$, $y$ and $z$ occur, only the symbols $*$ and $\wedge$ are allowed; 
\item basic quantification rules (cf. \cite[\S~3.5]{Troelstra1996}), i.e., \medskip
\begin{itemize}
\item $\inferrule{\alpha(\varphi[a/v]\wedge \forall v. \varphi) \deds \psi}{\alpha(\forall v. \varphi) \deds \psi}$ \quad and \quad $\inferrule{\varphi \deds \psi[a/v]}{\varphi \deds \exists v. \psi}$; \medskip
\item $\inferrule{\varphi \deds \psi}{\varphi \deds \forall v. \psi}$ \, and \, $\inferrule{\alpha(\psi) \deds \varphi}{\alpha(\exists v.\psi) \deds \varphi}$ whenever $v$ is fresh for $\varphi$ and $\alpha$; \medskip
\end{itemize}
\item while we have not needed such axioms and rules in the derivations presented so far, any  system used for reasoning about simple programs is likely to need additional principles governing pointer axioms---at the very least, some variant of a rule encoding disjointness of heaps: 
$$
\inferrule{ \alpha(e \mapsto a_1 * e \mapsto a_2) \deds \varphi}{\alpha(\bot) \deds \varphi} 
$$
and a rule encoding functionality of pointers:
$$
\inferrule{ \alpha(e \mapsto a_1 \wedge e \mapsto a_2) \deds \varphi}{\alpha(a_1 = a_2 \wedge e \mapsto a_1) \deds \varphi}; 
$$
\item finally, as a parameter in the definition of our proof system, we allow the user to choose a bunched-sequent-style formulation  of a fragment of arithmetic with good proof-theoretic  properties, like Skolem's primitive recursive arithmetic (PRA) \cite[\S~4.5.2]{Troelstra1996} or a chosen fragment of Presburger's arithmetic. Note again that it makes perfect sense to work with fragments which do not allow unrestricted pointer arithmetic, thus removing the need for incorporating arithmetic in our proof system. All we needed for inference rules and derivations so far was the ability to encode $e \mapsto_{\ell} \overline{a}$ and this we could do having only the syntax and axioms of the successor function; as discussed in \S~\ref{sec:DecRev} below one can go still further than that, take each $e \mapsto_{\ell} \overline{a}$ to be an atom in its own right and even restrict the length of $\overline{a}$ in such an expression.  Even with no arithmetic present, when one is taking BI rather than BBI as the propositional base, 
 it is natural to enrich the system so that one can derive the law of excluded middle for equality statements. 
\end{itemize}
\end{itemize}


\noindent
Note here that 
  there are candidates for rules which can be admissible, but not necessarily derivable. O'Hearn et al. \cite[\S~3.2.1]{OHearnYR09:acm} give  as an example what they call the \emph{the conjunction rule}:

$$
\inferrule{\{ P_1 \} \, C  \, \{ R_1 \} \quad \{ P_2 \} \, C  \, \{ R_2 \}}{\{ P_1 \wedge P_2 \} \, C  \, \{ R_1 \wedge R_2\}}.
$$
Let us observe that in the presence  of the consequence rule this rule is clearly suboptimally formulated, just like \emph{auxiliary variable renaming} or \emph{auxiliary variable elimination} \cite{Reynolds02:lics,YangPhD} derivable in our system. If one wants this rule to be derivable, it is enough to add
$$
\inferrule{\{ P \} \, C  \, \{ R_1 \} \quad \{ P \} \, C  \, \{ R_2 \}}{\{ P \} \, C  \, \{ R_1 \wedge R_2\}}.
$$
This indicates a more general pattern of rules for assertions mimicking sequent-style rules, where the command itself plays a r\^ole similar to that of a turnstile (inequality) sign. \textsc{VarEl$\exists$} and \textsc{VarEl$\forall$} above follow the same pattern. Yet another one, also sound over the intended  semantics, would be 
$$
\inferrule{\{ P_1 \} \, C  \, \{ R \} \quad \{ P_2 \} \, C  \, \{ R \}}{\{ P_1 \vee P_2\} \, C  \, \{ R \}}.
$$
It is not immediately obvious  whether all such rules are admissible in the proposed system. Note that the restrictions necessary to ensure soundness of the frame rule indicate that only translations of the rules governing additive connectives are worth considering in this context.

\subsection{Decidability Revisited} \label{sec:DecRev}


We have already noted in \S~\ref{HoareIntro} that whenever the assertion language  contains arithmetic (or anything sufficiently rich to encode it), G\"odel's Incompleteness Theorem obviously implies that the set of valid of assertions cannot be even recursively enumerable. 
  As we have discussed above, one sensible strategy 
    is to focus on incomplete proof systems or decision procedures---and there is no shortage of useful heuristics.  But we have also indicated that especially when  reasoning about typical operations on shared mutable data structures---like linked list reversal or copying/deletion of a tree---one hardly ever needs full pointer arithmetic. Consequently, one can deal with assertion languages which are not automatically covered by G\"odel's result. To  improve the situation even more, one can further restrict the assertion language, e.g., by limiting the number of quantified variables. The limit case is reached when there are no quantified variables left: all the expressions of the assertion language constructed without the use of multiplicative (a.k.a. spatial, separating or simply substructural)  connectives can be used verbatim as guards of WHILE or IF expressions of the programming language.

It would seem that such a propositional setup is precisely the one we have considered in \S~\ref{Decidability}, hence undecidability limits discussed therein 
 still apply. But an astute reader may have already recalled Remark \ref{rem:undecvar}: the absence of uninterpreted algebraic variables standing for arbitrary propositions limits direct applicability of such purely algebraic results. And indeed, expressions of the assertion language of SL are built from concrete atoms of the form $a_1 = a_2$, $emp$ or $e \mapsto_{\ell} \overline{a}$ (note again that if we do not assume that the assertion language can directly encode at least the successor function, we need to allow a more general form of pointer atoms). How do we know that in at least some of the simpler boolean allocation/heap models of \S\S~\ref{Effect} and \ref{HeapModel} such a restriction does not rule out valuations crucial for establishing undecidability of the set of BBI formulas valid in that model?

As it turns out, this is precisely what happens. 
 Calcagno, Yang and O'Hearn \cite[\S\S~4--5]{CalcagnoYH01}  show that the quantifier-free BBI language obtained by restricting the pointer atoms to the binary form $e \mapsto_\ell a_1, a_2$ (and with no function symbols) interpreted over store-and-heap models where $RV = L \cup \{ \nil \}$  and
  heaps are finite partial functions from $L$ to $RV \times RV$ 
  is PSPACE-complete, with further restrictions allowing even better complexity. By contrast, the set of all valid assertions in the quantified version of this language is not even recursively enumerable \cite[Th. 1]{CalcagnoYH01}. A detailed discussion of this phenomenon is provided by Brotherston and Kanovich \cite[\S~10]{BrotherstonK14:jacm}, who show how restriction to \emph{finite valuations} in heap models can restore decidability for the propositional language.\footnote{Note that 
   the denotation of pointer atoms is  finite only if both $\textit{PVar}$ and $\textit{Val}$ are finite.  Brotherston and Kanovich \cite[\S~10]{BrotherstonK14:jacm} circumvent this by stating corresponding theorems in the heap-only setting, but given that the original result of Calcagno et al. \cite{CalcagnoYH01} was proved for a store-and-heap model, a somewhat more general formulation would be desirable.} Demri and Deters \cite[\S~4.3.2]{DemriD15:jancl} note that an analogous PSPACE-completeness result holds with pointer atoms of the form $e \mapsto_\ell a_1, \dots, a_k$ for arbitrary but fixed finite $k$. Furthermore, Demri et al. \cite{Demri2014} show that when only pointer atoms of the form $e \mapsto a_1$ are allowed, the PSPACE upper bound survives in the presence of just one quantified variable (however, this result cannot be combined with atoms of the form $e \mapsto_\ell a_1, a_2$ \cite{DemriD15:acm}). 
 For more positive and negative decidability results for various fragments of the assertion language, the reader is referred to the overview of  Demri and Deters \cite{DemriD15:jancl}.


\section{Bi-Abduction: The Main Issue of SL Proof Theory} \label{sec:biabduction}

The story of algorithmic questions dealt with by Separation Logic would be incomplete if we finished it here. We are now in a position to briefly discuss 
 perhaps the most important proof-theoretic tasks for SL practitioners, which may be somewhat novel for more traditionally oriented algebraists and logicians. In the words of Peter O'Hearn (p.c.),
\begin{quote}
[t]he one thing I wish we could get across to substructural logicians is the importance of inference questions beyond validity. Chief among these are frame inference and abduction. \cite{OHearn17:email}
\end{quote}
To be sure, \emph{abductive inference} is not exactly an unknown concept in philosophy and logic, its study dating back to Charles S. Peirce, with  the term
 ``used in two related but different senses'' \cite{sep-abduction} regarding the use of explanatory reasoning in either \emph{generating} or \emph{justifying} hypotheses. As stressed by, e.g., the corresponding entry in the \emph{Stanford Encyclopedia of Philosophy} \cite{sep-abduction}, contemporary philosophers of science tend to employ it in the latter meaning (\emph{context of justification} or \emph{inference to the best explanation}), whereas Peirce himself\footnote{Although when it comes to Peirce's own views, ``[i]t is a common complaint that no coherent picture emerges from Peirce's writings on abduction. (Though perhaps this is not surprising, given that he worked on abduction throughout his career, which spanned a period of more than fifty years \dots )'' \cite{sep-abduction}.} put it in the context of discovery:
 \begin{quote}
 Abduction is the process of forming explanatory hypotheses. It is the only logical operation which introduces any new idea (...) Abduction must cover all
the operations by which theories and conceptions are engendered. \cite[CP 5.172,5.590]{Peirce}
 \end{quote}
 It seems safe to say that the meaning of the term as used in  computer science and artificial intelligence \cite{Paul93:air,DeneckerK01:jlp,CalcagnDOHY11:acm} either combines the context of discovery  with that of  justification or even focuses specifically on the  former one, thus being closer to original concerns of Peirce. 
 
 What does exactly abduction and \emph{bi-abduction} \cite{CalcagnDOHY11:acm} consist in? Below, we propose two formulations:  a general algebraic one (revealing a connection with unification) and a more specific one, sticking closely to both the proof system proposed in \S~\ref{sec:proofsl} and the paper of Calcagno et al. \cite{CalcagnDOHY11:acm}.
 
 \subsection{Abduction and Bi-Abduction Algebraically} \label{BiabductionAlg}
 
  Algebraically, one may think of the problem of abduction as follows: given
  \begin{itemize}
  \item a formal language $\lL$ and a theory $T$ in $\lL$ whose models all include as subreducts ordered monoids, with  $\leq$ being the ordering (either primitive or  term-definable)  and $\cdot$ being the semigroup operation (again, either primitive or  term-definable), 
  \item two terms $h, c \in \lL$ (called, respectively, the \emph{hypothesis} and the \emph{conclusion}), 
  \item a set of \emph{potential antiframes} (relative to $T$, $h$ and $c$) $Fr^-(T,h,c) \subseteq \lL$,  
  \end{itemize}
  
  \begin{center}
  find $a \in Fr^-(T,h,c)$ s.t. $T \vdash h \cdot a \leq c$.  
  \end{center}
\noindent  
Whenever $\lL$ and $T$ yield (either primitive or term-definable) left-residual $\backslash$ of~$\cdot$, only the presence of $Fr^-(T,h,c)$ prevents the problem from collapsing into triviality: otherwise, one could always take $a = h \backslash c$, and an even more dramatic trivialization  would be possible whenever $T \vdash h \cdot \bot = \bot$. Furthermore, whereas traditional forms of abduction involve $\cdot$ being the additive multiplication $\wedge$, in the context of SL one is naturally interested in the \emph{spatial} abduction, with ``$\cdot$'' being ``$*$''. 
Perhaps most importantly, however, from the point of view of concerns of SL, a more general (and symmetric) problem is of more interest. Calcagno et al. \cite{CalcagnDOHY11:acm} baptised it \emph{bi-abduction}. 

Apart from taking as input data the same  $\lL$, $T \subseteq \lL$, $h, c \in \lL$ and $Fr^-(T,h,c) \subseteq \lL$, the problem of bi-abduction also requires  \emph{potential frames} (relative to $T$, $h$ and $c$) $Fr^+(T,h,c) \subseteq \lL$; needless to say, it can well happen that $Fr^-(T,h,c) = Fr^+(T,h,c)$. The problem is then to 

 \begin{center}
  find $a \in Fr^-(T,h,c)$ and $f \in Fr^+(T,h,c)$ s.t. $T \vdash h \cdot a \leq c \cdot f$.  
  \end{center}

\begin{rem}
Especially in the presence of semi-lattice connectives like $\wedge$, this general statement of bi-abduction can be reformulated as a special case of a \emph{restricted unification problem} \cite{Burckert1986} (\emph{modulo theory}). Namely, given \emph{fresh} $x$, $y$, the challenge is to find a substitution $\sigma$ defined on $\{x, y\}$ (i.e., leaving other variables unchanged)  s.t. $\sigma x \in  Fr^-(T,h,c)$, $\sigma y \in  Fr^+(T,h,c)$ and $\sigma(h \cdot x \wedge c \cdot y) = \sigma(h \cdot x)$. We leave the exploration of this perspective for future investigation.
\end{rem}



\subsection{Bi-Abduction in Separation Logic} \label{BiabductionSL}

The above presentation of abduction and bi-abduction is much more general than the challenge  of Calcagno et al. \cite{CalcagnDOHY11:acm}, which can be formulated concretely in terms of the proof system sketched in \S~\ref{sec:proofsl}: given $H$ and $C$, find \emph{antiframe} $\alpha$ and \emph{frame} $\varphi$ s.t.
$$H * \alpha \deds C * \varphi,$$
where not only $\alpha$ and $\varphi$, but also $H$ and $C$ themselves are \emph{symbolic heaps} of the form $\exists \overline{v}. \Pi \wedge \Sigma$, $\Pi$ being a \emph{pure formula} and $\Sigma$ being a \emph{spatial formula} defined as follows:
\begin{align*}
\Pi, \Pi' & ::= a_1 = a_2 \mid a_1 \neq a_2 \mid \top \mid \Pi \wedge \Pi' \\
\Sigma, \Sigma' & ::= a_1 \mapsto a_2 \mid emp \mid \top \mid \Sigma * \Sigma'.
\end{align*}
Furthermore, as already discussed in \S~\ref{sec:DecRev}, there is no reason to insist on $a_1$ and $a_2$ being entirely arbitrary arithmetical expresssions. In fact, one often can restrict them to being elements of $PVar$, $AVar$ plus a suitable collection of additional constants. On the other hand, as we also discussed in  \S~\ref{sec:DecRev}, one might often need a richer collection of spatial predicates, at the very least replacing $e \mapsto a$ with $e \mapsto_\ell \overline{a}$ and possibly more (cf., e.g., \emph{abstract predicates} of Parkinson and Bierman \cite{Parkinson2005}). 
 A well-behaved class of similar formulas is  the ``Smallfoot fragment'' (cf. \S~\ref{sec:tools}) as defined by Demri and Dieters \cite[\S~4.3.1]{DemriD15:jancl}.

Just like in \S~\ref{BiabductionAlg}, \emph{abduction} is a problem with the same input as bi-abduction,  but the task is just to find \emph{antiframe} $\alpha$  s.t.
$H * \alpha \deds C$.  Calcagno et al.  \cite{CalcagnDOHY11:acm} provide an analysis of 
 minimality and termination of proof search for abduction in this setting, and a somewhat more sketchy one for bi-abduction, leaving a more throughout discussion of theoretical issues involved for future work. In \S~\ref{sec:tools} below, we are going to say a few more words about practical importance of (bi-)abduction for concrete tools.

\section{Applications and Later Developments} \label{sec:algsep}

In this section, we are going to briefly discuss 
applications, generalizations and developments which we cannot present in detail in this overview. 

\subsection{Competing Formalisms} \label{sec:competing}

It would not be adequate to claim that SL has had no competitors to solve the
 problems plaguing Hoare reasoning about pointer programs presented at the beginning of \S~\ref{SepSpecific}. 
  Bornat \cite{Bornat2000} provides an overview of the state of the art exactly at the time when SL entered the scene.

Later, Kassios \cite{Kassios06:fm} 
 suggested  another, object-oriented  alternative 
  in the form of the theory of \emph{dynamic frames}  (concerning the name, recall Footnote \ref{ft:frameorigin}). Soon afterwards, the theory of \emph{implicit dynamic frames} \cite{Smans2009} rather successfully combined the insights of dynamic frames with those of SL. 
    In particular,  this has led to the continuing development of the tool VeriFast \cite{Jacobs2011,Philippaerts2014}, whose core theory has been moreover formalized in the Coq proof assistant \cite{JacobsVP15:lmcs}. One of most important features inherited by implicit dynamic frames from separation logic is the presence of 
   $*$ in the assertion language. 

\subsection{Tools} \label{sec:tools}

VeriFast \cite{Jacobs2011,JacobsVP15:lmcs,Philippaerts2014} mentioned above is just one example of a recent, industrial-strength tool incorporating separation logic insights. But the story of such tools begins with Smallfoot \cite{Berdine2006}. Its invention was preceded by investigation of decision procedures for fragments of the assertion language even better behaved than those appearing in \S~\ref{sec:DecRev} \cite{Berdine2005:fsttcs} (cf. also the discussion of the ``Smallfoot fragment'' in Demri and Deters \cite{DemriD15:jancl} and in \S~\ref{BiabductionSL} above) and \emph{symbolic execution} in separation logic context \cite{Berdine2005:aplas}. Another paper published at the same time which greatly contributed to subsequent popularity of SL and formalisms utilizing the frame rule 
 was the work of Parkinson and Bierman \cite{Parkinson2005} introducing \emph{abstract predicates}.

Subsequently, the SL community produced more automated tools like Space\-Invader \cite{Yang2008}, SLAyer \cite{Berdine2011} at Microsoft Research\footnote{\url{https://github.com/Microsoft/SLAyer}} (see also \cite{Berdine2005:aplas,Distefano2006} for underlying research on symbolic execution) and, especially, the static analyser Infer \cite{Calcagno2011}, presently developed at Facebook \cite{Calcagno2015:nasa}, but available open-source\footnote{\url{https://github.com/facebook/infer}}.  Infer crucially relies on frame inference and bi-abduction discussed in \S~\ref{sec:biabduction}.

\begin{quote}
So, a substructural logic is used in a tool that prevents thousands of bugs per month from reaching production in products used by over 1 billion people daily. 
 \cite{OHearn17:email}
\end{quote}


\noindent
Given the inherent computational limitations for fully algorithmic solutions, however, approaches based on proof assistants seem a natural alternative option, especially from an academic perspective. 
 While there exists work on encoding separation logic, e.g., in Isabelle/HOL \cite{Tuerk11}, 
 Coq seems the most common setting  for such developments.
  Coq verification of Featherweight VeriFast \cite{JacobsVP15:lmcs} illustrates that proof assistants may have a r\^ole to play even with fully automated tools.  Another  recent Coq-based line of work is a series of frameworks such as  ModuRes \cite{Sieczkowski2015}, Iris 2.0 \cite{Jung2016}, and MoSeL \cite{KrebbersEA18:icfp} with   theoretical underpinnings in higher-order 
   BI-\emph{hyperdoctrines} \cite{BieringBT07}. Finally, separation logic is also being gradually incorporated in Coq-based courses \cite{Chlipala2016,Paulus2016,PierceSF,Litak18:sf}. 
 

\subsection{Concurrency and Algebraic Aspects} \label{sec:concur}

The rich collection of models discussed in \S\S~\ref{ModelsBI}--\ref{ModelsGBI} suggests that BI leads to 
 more applications than reasoning about pointer programs in sequential separation logic. Of all such developments, we most regret not being able to devote more attention in this overview to \emph{concurrent separation logic}. We can only refer the reader to a recent overview by Brookes and O'Hearn \cite{BrookesOH2016}, which was written following the award of the 2016 G\"odel Prize  to both authors for their involvement in this formalism \cite{Brookes2007,OHearn2007:tcs}. The very least we should say is that most tools and frameworks mentioned in \S~\ref{sec:tools} allow reasoning about concurrent programs. On the theoretical front, we only touched upon relevant issues in Remark \ref{rem:semframe}.

A  development closely related to concurrent separation logic whose omission we particularly regret is \emph{concurrent Kleene algebra} (\emph{CKA}) \cite{HoareMSW11,OHearnPVH15}. And this is a good opportunity to finish by returning to the main algebraic theme of this overview.
   While equational features of Floyd-Hoare logics have been noticed and substantially used in monographs oriented towards category theory, like Manes and Arbib  \cite{ManesArbib} or Bloom and \'Esik \cite{Bloom93}, a good reference for a more traditional algebraist  is provided by Kozen \cite{Kozen2000}  showing how to encode Floyd-Hoare logics in  \emph{Kleene algebra with tests} (\emph{KAT}; for an important predecessor see, e.g., Pratt \cite{Pratt1976} discussing the relationship between Floyd-Hoare, Tarskian and modal semantics). It remains to be seen whether  \emph{concurrent Kleene algebra with tests} (\emph{CKAT}, \cite{Jipsen14,JipsenM16}) is going to play a comparably important r\^ole. There is also an alternative algebraic approach to separation logic based on quantales \cite{DangHM11}. 

\begin{acknowledgement}
We would like to thank: \emph{Hiroakira Ono}, without whom both authors would not have met once upon a time in western Japan, there would have been no stimulus to write this overview, and many other things would not have happened; \emph{Nick Galatos} and \emph{Kazushige Terui} for suggesting the idea to write this overview, 
 and for their patience and support during the very long write-up period; Nick, \emph{Peter O'Hearn}, \emph{Revantha Ramanayake} and \emph{Simon Docherty} for their comments in the final stages of write-up, 
 in Peter's case including the suggestion to add some material on bi-abduction (\S~\ref{sec:biabduction}) and feedback regarding Infer and automated tools discussed in \S~\ref{sec:tools}. 
  Moreover, the second author wishes to thank: the family of the first author, in particular \emph{Julie Tapp}, for hosting him for two weeks in April 2015, when the bulk of this paper was written;  his project student \emph{Dominik Paulus} for developing a convenient Coq formalization \cite{Paulus2016}, which proved helpful when working on \S\S~\ref{SeparationLogic}--\ref{sec:sltheory}; and \emph{Erwin R. Catesbeiana}, for displaying a tantalizing view on the empty heaplet. 
\end{acknowledgement}


\printbibliography

\end{document}